\documentclass[sigconf, nonacm]{acmart}

\usepackage{subcaption}
\usepackage[export]{adjustbox}
\usepackage{multirow}
\usepackage{xspace}
\usepackage{xcolor}
\usepackage{algorithmic}
\algsetup{linenosize=\small}
\usepackage[ruled,lined,linesnumbered,noend]{algorithm2e}
\usepackage{enumeitem}
\usepackage{amsmath}
\usepackage{fancyvrb}
\newcommand\vldbdoi{XX.XX/XXX.XX}
\newcommand\vldbpages{XXX-XXX}
\newcommand\vldbvolume{14}
\newcommand\vldbissue{1}
\newcommand\vldbyear{2020}
\newcommand\vldbauthors{\authors}
\newcommand\vldbtitle{\shorttitle} 
\newcommand\vldbavailabilityurl{URL_TO_YOUR_ARTIFACTS}
\newcommand\vldbpagestyle{plain} 
\newcommand\ignore[1]{}

\begin{document}


\newcommand{\mc}[1]{\mathcal{#1}}
\newcommand{\nit}[1]{\textit{#1}}
\newcommand{\mostafa}[1]{\textcolor{red}{Mostafa: #1}}
\newcommand{\bl}[1]{\textbf{#1}}
\newcommand{\HL}[1]{\framebox[1.1\width]{\textbf{#1}}}
\newcommand{\starter}[1]{\vspace{0.1mm}\noindent \textbf{#1}}
\newcommand{\true}{\texttt{true}\xspace}
\newcommand{\SAlg}{\textit{SimpleTerm}\xspace}
\newcommand{\LAlg}{\textit{LinearTerm}\xspace}
\newcommand{\dsimplify}{\textit{DynSimplify}\xspace}
\newcommand{\simplify}{\textit{Simplify}\xspace}
\newcommand{\false}{\texttt{false}\xspace}
\newcommand{\tshape}{\texttt{t-shapes}\xspace}
\newcommand{\tgraph}{\texttt{t-graph}\xspace}
\newcommand{\tparse}{\texttt{t-parse}\xspace}
\newcommand{\tcomp}{\texttt{t-comp}\xspace}
\newcommand{\nshape}{\texttt{n-shapes}\xspace}
\newcommand{\npred}{\texttt{n-predicates}\xspace}
\newcommand{\nfact}{\texttt{n-atoms}\xspace}
\newcommand{\nrule}{\texttt{n-rules}\xspace}
\newcommand{\ttotal}{\texttt{t-total}\xspace}
\newcommand{\deepone}{\texttt{deep-100}\xspace}
\newcommand{\deeptwo}{\texttt{deep-200}\xspace}
\newcommand{\deepthree}{\texttt{deep-300}\xspace}
\newcommand{\lubmone}{\texttt{LUBM-1}\xspace}
\newcommand{\lubmtwo}{\texttt{LUBM-10}\xspace}
\newcommand{\lubmthree}{\texttt{LUBM-100}\xspace}
\newcommand{\lubmfour}{\texttt{LUBM-1K}\xspace}
\newcommand{\ibenchont}{\texttt{ONT-256}\xspace}
\newcommand{\ibenchstd}{\texttt{STB-128}\xspace}
\newcommand{\ibench}{iBench\xspace}
\newcommand{\edb}{EDB\xspace}
\newcommand{\shapes}{\nit{Shapes}\xspace}
\newcommand{\applicable}{\nit{Applicable}\xspace}


\newcommand{\from}[3]{\textbf{FROM} #1 \textbf{TO} #2: #3}
\def\FP{\text{\rm FP}}
\def\spanL{\text{\rm SpanL}}
\def\spanLL{\text{\rm SpanLL}}
\def\sharpP{\text{\rm \#P}}
\def\sharpL{\text{\rm \#L}}
\def\NL{\text{\rm NL}}
\def\PTIME{\text{\rm P}}
\def\PH{\text{\rm PH}}
\def\NP{\text{\rm NP}}
\def\LOGSPACE{\text{\rm L}}
\def\NSPACE{\rm NSPACE}
\def\co{\rm co\text{-}}
\def\PSPACE{\rm PSPACE}
\def\EXPTIME{\rm EXPTIME}
\def\TWOEXPTIME{\rm 2EXPTIME}
\def\AEXSPACE{\rm AEXSPACE}
\def\ACZ{\rm AC_0}
\def\hard{\rm \text{-}hard}
\def\complete{\text{-{\rm complete}}}

\newcommand{\depth}[1]{\mathsf{depth}(#1)}
\newcommand{\rank}[1]{\mathsf{rank}(#1)}
\newcommand{\sdepth}[1]{\star\text{-}\mathsf{depth}(#1)}
\newcommand{\mh}[1]{\mathsf{mh}(#1)}
\newcommand{\ma}[1]{\mathsf{ma}(#1)}
\newcommand{\md}[1]{\mathsf{md}(#1)}
\newcommand{\smd}[1]{\star\text{-}\mathsf{md}(#1)}
\newcommand{\omd}[1]{\obl\text{-}\mathsf{md}(#1)}
\newcommand{\gforest}[1]{\mathsf{gforest}(#1)}
\newcommand{\gtree}[1]{\mathsf{gtree}(#1)}
\newcommand{\slinpath}[1]{simple linear $#1$-path}
\newcommand{\linpath}[1]{linear $#1$-path}
\newcommand{\guardedpath}[1]{guarded $#1$-path}
\newcommand{\stickypath}[1]{sticky $#1$-path}
\newcommand{\R}{\mathcal{R}}
\newcommand{\Lang}{\mathcal{L}}
\newcommand{\mi}[1]{\mathit{#1}}
\newcommand{\ins}[1]{\mathbf{#1}}
\newcommand{\adom}[1]{\mathsf{dom}(#1)}
\newcommand{\ra}{\rightarrow}
\newcommand{\fr}[1]{\mathsf{fr}(#1)}
\newcommand{\dep}{\Sigma}
\newcommand{\sch}[1]{\mathsf{sch}(#1)}
\newcommand{\sign}{\ins{S}}
\newcommand{\body}[1]{\mathsf{body}(#1)}
\newcommand{\head}[1]{\mathsf{head}(#1)}
\newcommand{\guard}[1]{\mathsf{guard}(#1)}
\newcommand{\class}[1]{\mathsf{#1}}
\newcommand{\pos}[1]{\mathsf{pos}(#1)}
\newcommand{\spos}[1]{\mathsf{spos}(#1)}
\newcommand{\app}[2]{\langle #1,#2 \rangle}
\newcommand{\tup}[1]{\langle #1 \rangle}
\newcommand{\crel}[1]{\prec_{#1}}
\newcommand{\tcrel}[1]{\prec_{#1}^{\star}}
\newcommand{\rctaa}{\class{CT}_{\forall \forall}^{\mathsf{res}}}
\newcommand{\rctaapr}{\mathsf{CT}_{\forall \forall}^{\mathsf{res}}}
\newcommand{\rctae}{\class{CT}_{\forall \exists}^{\mathsf{res}}}
\newcommand{\rctaepr}{\mathsf{CT}_{\forall \exists}^{\mathsf{res}}}
\newcommand{\base}[1]{\mathsf{base}(#1)}
\newcommand{\eqt}[1]{\mathsf{eqtype}(#1)}
\newcommand{\var}[1]{\mathsf{var}(#1)}
\newcommand{\const}[1]{\mathsf{const}(#1)}
\newcommand{\result}[2]{\mathsf{result}(#1,#2)}
\newcommand{\soresult}[2]{\sobl\text{-}\mathsf{result}(#1,#2)}
\newcommand{\oresult}[2]{\obl\text{-}\mathsf{result}(#1,#2)}
\newcommand{\sresult}[2]{\star\text{-}\mathsf{result}(#1,#2)}

\newcommand{\reach}[1]{\rightsquigarrow_{#1}}
\newcommand{\obl}{\mathsf{o}}
\newcommand{\sobl}{\mathsf{so}}
\newcommand{\std}{\mathsf{std}}
\newcommand{\cta}[1]{\class{CT}_{\forall \forall}^{#1}}
\newcommand{\cte}[1]{\class{CT}_{\forall \exists}^{#1}}

\newcommand{\ctda}[2]{\class{CT}_{\forall,#2}^{#1}}
\newcommand{\ctde}[2]{\class{CT}_{\exists,#2}^{#1}}

\newcommand{\ct}[1]{\class{CT}_{\forall}^{#1}}
\newcommand{\ctd}[2]{\class{CT}^{#1}_{#2}}
\newcommand{\ctapr}[1]{\mathsf{CT}_{\forall \forall}^{#1}}
\newcommand{\ctepr}[1]{\mathsf{CT}_{\forall \exists}^{#1}}
\newcommand{\ctdapr}[1]{\mathsf{CT}_{\forall}^{#1}}
\newcommand{\ctdepr}[1]{\mathsf{CT}_{\exists}^{#1}}
\newcommand{\ctpr}[1]{\mathsf{CT}_{\forall}^{#1}}
\newcommand{\ctdpr}[1]{\mathsf{CT}^{#1}}
\newcommand{\cri}[1]{\mathsf{cr}(#1)}
\newcommand{\lin}[1]{\mathsf{Lin_{\class{S}}}(#1)}
\newcommand{\ling}[1]{\mathsf{lin}(#1)}
\newcommand{\shape}[1]{\mathsf{shape}(#1)}
\newcommand{\svar}[1]{\mathsf{svar}(#1)}
\newcommand{\constfree}[1]{\mathsf{c\text{-}free}(#1)}
\newcommand{\id}[2]{\mathsf{id}_{#1}(#2)}
\newcommand{\unique}[1]{\mathsf{unique}(#1)}
\newcommand{\simple}[1]{\mathsf{simple}(#1)}
\newcommand{\dsimple}[2]{\mathsf{simple}_{#1}(#2)}
\newcommand{\gsimple}[1]{\mathsf{gsimple}(#1)}
\def\sub{\sqsubseteq}
\def\substrict{\sqsubset}

\newcommand{\norm}[1]{\mathsf{Norm}(#1)}
\newcommand{\depg}[1]{\mathsf{dg}(#1)}
\newcommand{\edepg}[1]{\mathsf{edg}(#1)}
\newcommand{\sodg}[1]{\mathsf{so\text{-}dg}(#1)}
\newcommand{\odg}[1]{\mathsf{o\text{-}dg}(#1)}
\newcommand{\ex}[1]{\mathsf{exvar}(#1)}
\newcommand{\mgu}[2]{\mathsf{mgu}(#1,#2)}
\newcommand{\res}[1]{\mathsf{res}(#1)}
\newcommand{\f}[2]{f_{#1}(#2)}
\newcommand{\arity}[1]{\mathsf{ar}(#1)}
\newcommand{\atoms}[1]{\mathsf{atoms}(#1)}
\newcommand{\birth}[2]{\mathsf{birth}_{#1}(#2)}
\newcommand{\OMIT}[1]{}
\newcommand{\crt}[1]{\texttt{cr}(#1)}
\newcommand{\precd}[1]{\prec_{#1}}
\newcommand{\lvl}[2]{\texttt{lv}_{#1}(#2)}
\newcommand{\posvar}[2]{\mathsf{pos}(#1,#2)}
\newcommand{\posterm}[2]{\mathsf{pos}(#1,#2)}
\newcommand{\varpos}[2]{\mathsf{var}(#1,#2)}
\newcommand{\termpos}[2]{\mathsf{term}(#1,#2)}
\newcommand{\CT}[2]{\mathsf{CT}^{#1}_{#2}}
\newcommand{\frpos}[1]{\mathsf{frpos}(#1)}
\newcommand{\dom}{\mathbf{C}}
\newcommand{\freshdom}{\mathbf{N}}
\newcommand{\frontier}[1]{\mathsf{fr}(#1)}
\newcommand{\eqtype}[1]{\mathsf{eqtype}(#1)}
\def\iso{\simeq}
\newcommand{\can}[1]{\mathsf{can}(#1)}
\newcommand{\proj}[2]{\Pi_{#1}(#2)}
\newcommand{\pred}[1]{\mathit{pred}(#1)}
\newcommand{\predt}[1]{[#1]}
\newcommand{\atom}[1]{\underline{#1}}
\newcommand{\tuple}[1]{\bar{#1}}
\newcommand{\resolv}[1]{[#1]}
\newcommand{\parent}[1]{\mathit{par}(#1)}
\newcommand{\dept}[1]{\mathit{depth}(#1)}
\newcommand{\chase}[2]{\mathsf{chase}(#1,#2)}
\newcommand{\chasesize}[2]{\mathsf{chsize}(#1,#2)}
\newcommand{\starchasei}[3]{\star\text{-}\mathsf{chase}^{#3}(#1,#2)}
\newcommand{\starchase}[2]{\star\text{-}\mathsf{chase}(#1,#2)}
\newcommand{\sochasei}[3]{\sobl\text{-}\mathsf{chase}^{#3}(#1,#2)}
\newcommand{\sochase}[2]{\sobl\text{-}\mathsf{chase}(#1,#2)}
\newcommand{\ochasei}[3]{\obl\text{-}\mathsf{chase}^{#3}(#1,#2)}
\newcommand{\ochase}[2]{\obl\text{-}\mathsf{chase}(#1,#2)}
\newcommand{\completion}[2]{\mathsf{complete}(#1,#2)}
\newcommand{\mar}[1]{\hat{#1}}
\newcommand{\nullobl}[3]{\bot^{#1}_{#2,#3}}
\newcommand{\nullsobl}[3]{\bot^{#1}_{#2,#3_{|\frontier{#2}}}}
\newcommand{\startype}[1]{\star\textrm{-}\mathsf{type}(#1)}
\newcommand{\type}[2]{\mathsf{type}_{#1}(#2)}
\newcommand{\types}[2]{#1\textrm{-}\mathsf{types}(#2)}
\newcommand{\src}[1]{\mathsf{src}(#1)}
\newcommand{\IDTGD}{\class{ID}}
\newcommand{\DLLITETGD}{\mathsf{DL\textrm{-}Lite^{TGD}}}
\newcommand{\SLTGD}{\class{SL}}
\newcommand{\LTGD}{\class{L}}
\newcommand{\GTGD}{\class{G}}
\newcommand{\WGTGD}{\class{WG}}
\newcommand{\RATGD}{\class{RA}}
\newcommand{\LARATGD}{\class{LARA}}
\newcommand{\LCRATGD}{\class{LCRA}}
\newcommand{\LCWATGD}{\class{LCWA}}
\newcommand{\WATGD}{\class{WA}}
\newcommand{\DAT}{\class{DAT}}
\newcommand{\UCQ}{\class{UCQ}}
\newcommand{\SLRATGD}{\class{SLRA}}
\newcommand{\SLWATGD}{\class{SLWA}}
\newcommand{\LRATGDP}{\class{LRA}^{+}}
\newcommand{\LWATGDP}{\class{LWA}^{+}}
\newcommand{\oblrew}[1]{\mathsf{enrichment}(#1)}

\newcommand{\termeq}{\equiv}
\newcommand{\skolem}[3]{\bot_{#1,#2}^{#3}}

\def\qed{\hfill{\qedboxempty}      
  \ifdim\lastskip<\medskipamount \removelastskip\penalty55\medskip\fi}

\def\qedboxempty{\vbox{\hrule\hbox{\vrule\kern3pt
                 \vbox{\kern3pt\kern3pt}\kern3pt\vrule}\hrule}}

\def\qedfull{\hfill{\qedboxfull}   
  \ifdim\lastskip<\medskipamount \removelastskip\penalty55\medskip\fi}

\def\qedboxfull{\vrule height 4pt width 4pt depth 0pt}

\newcommand{\markfull}{\qedboxfull}
\newcommand{\markempty}{\qed}


\title{Semi-Oblivious Chase Termination for Linear Existential Rules: An Experimental Study}

\author{Marco Calautti}
\affiliation{%
	\institution{University of Milan}
	\country{}
}
\email{marco.calautti@unimi.it}

\author{Mostafa Milani}
\affiliation{%
	\institution{University of Western Ontario}
	\country{}
}
\email{mostafa.milani@uwo.ca}

\author{Andreas Pieris}
\affiliation{%
	\institution{University of Edinburgh \&}
	\country{}
}
\affiliation{%
	\institution{University of Cyprus}
	\country{}
}
\email{apieris@inf.ed.ac.uk}

\begin{abstract}
The chase procedure is a fundamental algorithmic tool in databases that allows us to reason with constraints, such as existential rules, with a plethora of applications. It takes as input a database and a set of constraints, and iteratively completes the database as dictated by the constraints. A key challenge, though, is the fact that it may not terminate, which leads to the problem of checking whether it terminates given a database and a set of constraints. In this work, we focus on the semi-oblivious version of the chase, which is well-suited for practical implementations, and linear existential rules, a central class of constraints with several applications. In this setting, there is a mature body of theoretical work that provides syntactic characterizations of when the chase terminates, algorithms for checking chase termination, precise complexity results, and worst-case optimal bounds on the size of the result of the chase (whenever is finite). Our main objective is to experimentally evaluate the existing chase termination algorithms with the aim of understanding which input parameters affect their performance, clarifying whether they can be used in practice, and revealing their performance limitations.
\end{abstract}

\maketitle

\pagestyle{\vldbpagestyle}
\begingroup\small\noindent\raggedright\textbf{PVLDB Reference Format:}\\
\vldbauthors. \vldbtitle. PVLDB, \vldbvolume(\vldbissue): \vldbpages, \vldbyear.\\
\href{https://doi.org/\vldbdoi}{doi:\vldbdoi}
\endgroup
\begingroup
\renewcommand\thefootnote{}\footnote{\noindent
This work is licensed under the Creative Commons BY-NC-ND 4.0 International License. Visit \url{https://creativecommons.org/licenses/by-nc-nd/4.0/} to view a copy of this license. For any use beyond those covered by this license, obtain permission by emailing \href{mailto:info@vldb.org}{info@vldb.org}. Copyright is held by the owner/author(s). Publication rights licensed to the VLDB Endowment. \\
\raggedright Proceedings of the VLDB Endowment, Vol. \vldbvolume, No. \vldbissue\ %
ISSN 2150-8097. \\
\href{https://doi.org/\vldbdoi}{doi:\vldbdoi} \\
}\addtocounter{footnote}{-1}\endgroup

\ifdefempty{\vldbavailabilityurl}{}{
\vspace{.3cm}
\begingroup\small\noindent\raggedright\textbf{PVLDB Artifact Availability:}\\
The source code, data, and/or other artifacts have been made available at \url{\vldbavailabilityurl}.
\endgroup
}

\section{Introduction}\label{sec:intro}

The \emph{chase procedure} (or simply chase) is a fundamental algorithmic tool that has been successfully applied to several database problems such as checking logical implication of constraints~\cite{BeVa84,MaMS79}, containment of queries under constraints~\cite{AhSU79}, computing data exchange solutions~\cite{FKMP05}, and ontological query answering~\cite{CaGL12}, to name a few.
The chase takes as input a database $D$ and a set $\dep$ of constraints, which, for this work, are \emph{existential rules} (a.k.a. {\em tuple-generating dependencies} (TGDs)) of the form	$\forall \bar x \forall \bar y \left(\phi(\bar x,\bar y) \ra \exists \bar z\, \psi(\bar x,\bar z)\right)$,
where $\phi$ (the body) and $\psi$ (the head) are conjunctions of relational atoms, and it produces an instance $D_\dep$ that is a {\em universal model} of $D$ and $\dep$, i.e., a model that can be homomorphically embedded into every other model of $D$ and $\dep$. Somehow $D_\dep$ acts as a representative of all the models of $D$ and $\dep$. This is the reason for the ubiquity of the chase in databases, as discussed in~\cite{DeNR08}. Indeed, many database problems can be solved by simply exhibiting a universal model.
%

\subsection{The Chase in a Nutshell}

Roughly, the chase adds new tuples to the database $D$ (possibly involving null values that act as witnesses for the existentially quantified variables), as dictated by the TGDs of $\dep$, and it keeps doing this until all the TGDs of $\dep$ are satisfied.
There are, in principle, three different ways for formalizing this simple idea, which lead to different versions of the chase procedure:

\medskip

\noindent
\textbf{Oblivious Chase.} The first one, which leads to the {\em oblivious chase}, is as follows: for each pair $(\bar t,\bar u)$ of tuples of terms from the instance $I$ constructed so far, apply a TGD $\sigma = \forall \bar x \forall \bar y \left(\phi(\bar x,\bar y) \ra \exists \bar z\, \psi(\bar x,\bar z)\right)$ if $\phi(\bar t,\bar u) \subseteq I$, and $\sigma$ has not been applied in a previous chase step due to the same pair $(\bar t,\bar u)$, and add to $I$ the set of atoms $\psi(\bar t,\bar v)$, where $\bar v$ is a tuple of new terms not occurring in $I$.

\medskip

\noindent
\textbf{Semi-Oblivious Chase.} The second one, which is a refinement of the oblivious chase, and it gives rise to the {\em semi-oblivious chase}, is as follows: for each pair $(\bar t,\bar u)$ of tuples of terms from the instance $I$ constructed so far, apply a TGD $\sigma = \forall \bar x \forall \bar y \left(\phi(\bar x,\bar y) \ra \exists \bar z\, \psi(\bar x,\bar z)\right)$ if $\phi(\bar t,\bar u) \subseteq I$, and $\sigma$ has not been applied in a previous chase step due to a pair of tuples $(\bar t,\bar u')$, where $\bar u$ and $\bar u'$ might be different, and add to $I$ the set of atoms $\psi(\bar t,\bar v)$, where $\bar v$ is a tuple of new terms not in $I$.
In other words, a TGD $\sigma$ of the above form is applied only once due to a certain witness $\bar t$ for the variables $\bar x$.

\medskip

\noindent
\textbf{Restricted Chase.} The third one, which leads to the {\em restricted} (a.k.a.~{\em standard}) {\em chase}, is as follows: for each pair $(\bar t,\bar u)$ of tuples of terms from the instance $I$ constructed so far, apply a TGD $\sigma = \forall \bar x \forall \bar y \left(\phi(\bar x,\bar y) \ra \exists \bar z\, \psi(\bar x,\bar z)\right)$ if $\phi(\bar t,\bar u) \subseteq I$, and there is no tuple $\bar u'$ of terms from $I$ such that $\psi(\bar t,\bar u') \subseteq I$, i.e., the TGD is not already satisfied, and add to $I$ the set of atoms $\psi(\bar t,\bar v)$, where $\bar v$ is a tuple of new terms not in $I$.

\medskip

Summing up, the key difference between the (semi-)oblivious and restricted versions of the chase procedure is that the former apply a TGD whenever the body is satisfied, while the latter applies a TGD if the body is satisfied but the head is not.

\subsection{Restricted vs. (Semi-)Oblivious Chase.}
%

It is not difficult to verify that the restricted chase, in general, builds smaller instances than the (semi-)oblivious one. In fact, it is easy to devise an example where, according to the restricted chase, none of the TGDs should be applied, while the (semi-)oblivious chase builds an infinite instance. Here is such an example taken from~\cite{CaPi21}:
	
\begin{example}
	Consider the database $D = \{R(a,a)\}$ and the TGD 
	\[
	\forall x \forall y (R(x,y)\ \ra\ \exists z\, R(z,x)).
	\] 
	The restricted chase will detect that the database already satisfies the TGD, while the (semi-)oblivious chase will build the instance
	\[
	\{R(a,a),R(\bot_1,a),R(\bot_2,\bot_1),R(\bot_3,\bot_2),\ldots\},
	\]
	where $\bot_1,\bot_2,\bot_3,\ldots$ are (labeled) nulls. \hfill\markfull
\end{example}
	
It is also easy to devise examples where the semi-oblivious chase does not apply any TGD, whereas the oblivious chase builds an infinite instance.
%
%
\ignore{	
\begin{example}\label{exa:so-vs-o}
	Consider the database $D = \{R(a,a)\}$ and the TGD 
	\[
	\forall x \forall y (R(x,y)\ \ra\ \exists z\, R(x,z)), 
	\]
	The semi-oblivious chase will build the instance $\{R(a,a),R(a,\bot)\}$, where $\bot$ is a null, whereas the oblivious chase will build the instance 
	\[
	\{R(a,a),R(a,\bot_1),R(a,\bot_2),R(a,\bot_3),\ldots\},
	\]
	where $\bot_1,\bot_2,\bot_3,\ldots$ are nulls. \hfill\markfull
\end{example}
}
It is generally agreed that the oblivious version of the chase, although a very useful theoretical tool, has no practical applications due to the fact that it infers a lot of redundant information, which in turn leads to very large instances that are very often infinite. 
Concerning the other variants of the chase, the restricted one has a clear advantage over the semi-oblivious one as it generally builds smaller instances. But, of course, this advantage does not come for free: at each step, the restricted chase has to check that there is no way to satisfy the head of the TGD at hand, and this can be very costly in practice.
It has been recently observed that for RAM-based implementations the restricted chase is the indicated approach since the benefit from producing smaller instances justifies the additional effort for checking whether a TGD is already satisfied; see, e.g.,~\cite{BKMMPST17,KrMR19}. However, as discussed in~\cite{BKMMPST17}, an RDBMS-based implementation of the restricted chase is quite challenging, whereas an efficient implementation of the semi-oblivious chase is feasible.
Hence, both the semi-oblivious and restricted versions of the chase are relevant tools for practical implementations.

\subsection{Linear TGDs and Chase Termination}

There are indeed efficient implementations of the semi-oblivious and restricted chase that allow us to solve central database problems by adopting a materialization-based approach~\cite{BKMMPST17,KrMR19,NPMHWB15,UKJDC18}.
Nevertheless, for this to be feasible in practice we need a guarantee that the chase terminates, which is not always the case.
This fact motivated a long line of research on the chase termination problem, that is, given a database $D$ and a set $\dep$ of TGDs, to check whether the semi-oblivious or restricted chase of $D$ with $\dep$ terminates.
It is known that, in general, this is an undecidable problem. This has been established in~\cite{DeNR08} for the restricted chase, and it was observed a year later in~\cite{Marn09} that the same proof shows undecidability also for the semi-oblivious chase.
The undecidability proof given in~\cite{DeNR08}, however, constructs a sophisticated set of TGDs that goes beyond existing well-behaved classes of TGDs that enjoy certain syntactic properties. This observation leads to the obvious question: is the chase termination problem algorithmically solvable whenever we focus on well-behaved classes of TGDs?

A well-behaved class of TGDs, which attracted a considerable attention due to its simplicity, and also the fact that it strikes a good balance between expressiveness and complexity, is that of linear TGDs proposed in~\cite{CaGL12}. A TGD is {\em linear} if it has only one atom in its body, whereas the head can be an arbitrary conjunction of atoms. Such a TGD is called {\em simple-linear} if each variable in its body occurs only once, whereas variables in its head can repeat without any restriction.
Although, at first glance, (simple-)linear TGDs may look very inexpressive, it turns out that they are powerful enough to express database integrity constraints, as well as ontological axioms. In particular, we know that referential integrity constraints (a.k.a. inclusion dependencies) that form a central class of constraints~\cite{AbHV95}, can be easily expressed as simple-linear TGDs. 
Moreover, the important ontology language DL-Lite$_R$~\cite{CDLL*07}, which is based on Description Logics and forms the logical underpinning of OWL 2 QL, one of the popular profiles of the W3C committee’s Web Ontology
Language (OWL) standard for ontology languages, can be easily embedded into the class of simple-linear TGDs.

The chase termination problem in the presence of (simple-)linear TGDs has been extensively studied the last few years. Concerning the semi-oblivious version of the chase, there is a mature body of theoretical work that provides syntactic characterizations of when the chase terminates based on suitable acyclicity notions, algorithms for checking chase termination, precise complexity results, and worst-case optimal bounds on the size of the chase instance (whenever is finite)~\cite{CaGP22}. On the other hand, for the restricted version of the chase, we only have a decidability result via an algorithm that runs in double-exponential time, under the assumption that the head of the linear TGDs consists of a single atom~\cite{LMTU19}.
This striking difference on the progress that has been achieved should be attributed to the fact that the chase termination problem is significantly more challenging in the case of the restricted chase.

\subsection{Main Objective}

Having a complete theoretical understanding of the semi-oblivious chase termination problem in the presence of (simple-)linear TGDs, the next step is to experimentally evaluate the proposed algorithms with the aim of understanding which input parameters affect their performance, clarifying whether they can be applied in a practical context, and revealing their performance limitations. This is precisely the main objective of this work. Note that we do not consider the restricted chase as this will be very premature due to the lack of a good theoretical understanding of the problem in question; the latter is the subject of an ongoing research activity.

From the chase termination literature, we can inherit two types of algorithms for the semi-oblivious chase termination problem in the presence of (simple-)linear TGDs, that is, {\em materialization-based} and {\em acyclicity-based}~\cite{CaGP22}, which can be described as follows:

\medskip

\noindent \textbf{Materialization-Based.}
The materialization-based algorithms exploit the existence of worst-case optimal bounds on the size of the chase instance (whenever is finite). In particular, given a database $D$ and a set $\dep$ of (simple-)linear TGDs, we have an integer $k_{D,\dep}$ such that the chase of $D$ with $\dep$ terminates iff the size of the chase instance (i.e., the number of its atoms) is at most $k_{D,\dep}$. This immediately leads to a conceptually simple chase termination algorithm: simply run the semi-oblivious chase of $D$ with $\dep$ and keep a counter for the number of generated atoms, and if the count exceeds $k_{D,\dep}$, then conclude that the chase does not terminate; otherwise, it does.

\medskip

\noindent \textbf{Acyclicity-Based.} On the other hand, the acyclicity-based algorithms exploit the syntactic characterizations of when the chase terminates via suitable acyclicity notions. In particular, given a database $D$ and a set $\dep$ of (simple-)linear TGDs, we know that the chase of $D$ with $\dep$ terminates iff the dependency graph of $\dep$ (a standard way of representing a set of TGDs as a graph, which is defined in Section~\ref{sec:semi}) does not contain a ``bad'' cycle, where a ``bad'' cycle witnesses the fact that during the chase of $D$ with $\dep$ we eventually fall in a cyclic chase derivation that leads to non-termination. This again leads to a conceptually simple chase termination algorithm: construct the dependency graph of $\dep$, and if it has a ``bad'' cycle, then conclude that the chase does not terminate; otherwise, it does.

\medskip

An exploratory analysis showed that the materialization-based algorithms are simply too expensive for a chase termination check. This is because the worst-case upper bounds on the size of the result of the chase from~\cite{CaGP22} are very large, and thus, the algorithms are forced, in general, to construct extremely  large instances before being able to safely recognize that the chase does not terminate.
On the other hand, we have observed that the acyclicity-based algorithms were reasonably efficient with a lot of room for optimizations and improvements. Therefore, towards our main objective, we focused our attention on the acyclicity-based algorithms.

\subsection{Main Outcome and Challenges} 
Our experimental analysis revealed that for simple-linear TGDs the primary parameter impacting the runtime of the acyclicity-based algorithm is the size of the input set of TGDs, whereas the size of the input database does not play any crucial role. Interestingly, the algorithm is very fast (in the order of seconds) even for large sets of simple-linear TGDs (with up to 100K TGDs).
Now, concerning the more interesting case of linear TGDs, our analysis showed that the acyclicity-based algorithm consists of two components that are of different nature. In particular, there is a database-dependent component, whose performance is solely impacted by the size of the database, and a database-independent component, whose runtime is primarily affected by the size of the set of TGDs.
Interestingly, the overall runtime of the algorithm is quite reasonable, which is a strong evidence that fast checking for the termination of the semi-oblivious chase in the case of linear TGDs is not an unrealistic goal. Note that most of the total end-to-end runtime of the algorithm is taken by the database-dependent component, which indicates that our future efforts should be focused on improving that component.

\medskip

\noindent
\textbf{Technical Challenges.} Towards the above outcome concerning the acyclicity-based algorithms, we had to overcome a couple of technical challenges that led to results of independent interest:

\begin{itemize}
	\item It would not be possible to obtain the above insightful conclusions by naively implementing the algorithms in question as this would lead to poor performance; this is discussed in Section~\ref{sec:alg}. Hence, we had to revisit and refine the theoretical algorithms from~\cite{CaGP22} in order to obtain algorithms that are amenable to efficient implementations. The low-level implementation details of the refined algorithms are in Section~\ref{sec:implementation}.
	
	\item In order to stress test the algorithms in question, we had to synthetically generate databases and sets of TGDs. However, as discussed in Section~\ref{sec:setting}, existing data and TGD generators are not suitable for our purposes as they do not allow us to tune certain parameters that are crucial for evaluating our chase termination algorithms. To this end, we developed our own data and TGD generators, and used them to carefully generate the databases and TGDs that have been employed in our experimental analysis.
\end{itemize}

\medskip

\noindent{\em The experimental infrastructure and the source code can be found at https://github.com/mostafamilani/chase-termination.}
\section{Preliminaries} \label{sec:prel}

We consider the disjoint countably infinite sets $\ins{C}$, $\ins{N}$, and $\ins{V}$ of {\em constants}, {\em (labeled) nulls}, and {\em variables}, respectively. We refer to constants, nulls and variables as {\em terms}. For an integer $n > 0$, we write $[n]$ for the set $\{1,\ldots,n\}$.

\medskip

\noindent 
\textbf{Relational Databases.} A {\em schema} $\ins{S}$ is a finite set of relation symbols (or predicates) with associated arity. We write $R/n$ to denote that $R$ has arity $n > 0$; we may also write $\arity{R}$ for the integer $n$.
A {\em (predicate) position} of $\ins{S}$ is a pair $(R,i)$, where $R/n \in \ins{S}$ and $i \in [n]$, that essentially identifies the $i$-th argument of $R$. We write $\pos{\ins{S}}$ for the set of positions of $\ins{S}$, that is, the set $\{(R,i) \mid R/n \in \ins{S} \text{ and } i \in [n]\}$.
An {\em atom} over $\ins{S}$ is an expression of the form $R(\bar t)$, where $R/n \in \ins{S}$ and $\bar t$ is an $n$-tuple of terms. A {\em fact} is an atom whose arguments consist only of constants.
For a variable $x$ in $\bar t = (t_1,\ldots,t_n)$, let $\posvar{R(\bar t)}{x} = \{(R,i) \mid t_i = x\}$. 
We write $\var{R(\bar t)}$ for the set of variables in $\bar t$. The notations $\posvar{\cdot}{x}$ and $\var{\cdot}$ extend to sets of atoms.
An {\em instance} over $\ins{S}$ is a (possibly infinite) set of atoms over $\ins{S}$ with constants and nulls. A {\em database} over $\ins{S}$ is a finite set of facts over $\ins{S}$. The {\em active domain} of an instance $I$, denoted $\adom{I}$, is the set of terms (constants and nulls) occurring in $I$. For a singleton instance $\{\alpha\}$, we simply write $\adom{\alpha}$ instead of $\adom{\{\alpha\}}$.

\medskip

\noindent
\textbf{Substitutions and Homomorphisms.}
A {\em substitution} from a set of terms $T$ to a set of terms $T'$ is a function $h : T \ra T'$. Henceforth, we treat a substitution $h$ as the set of mappings $\{t \mapsto h(t) \mid t \in T\}$.
The restriction of $h$ to a subset $S$ of $T$, denoted $h_{|S}$, is the substitution $\{t \mapsto h(t) \mid t \in S\}$.
A {\em homomorphism} from a set of atoms $A$ to a set of atoms $B$ is a substitution $h$ from the set of terms in $A$ to the set of terms in $B$ such that $h$ is the identity on $\ins{C}$, and $R(t_1,\ldots,t_n) \in A$ implies $h(R(t_1,\ldots,t_n)) =  R(h(t_1),\ldots,h(t_n)) \in B$.

\medskip

\noindent
\textbf{Tuple-Generating Dependencies.} A {\em tuple-generating dependency} (TGD) $\sigma$ is a (constant-free) sentence
$
\forall \bar x \forall \bar y \left(\phi(\bar x,\bar y) \ra \exists \bar z\, \psi(\bar x,\bar z)\right),
$
where $\bar x, \bar y$ and $\bar z$ are tuples of variables of $\ins{V}$, and $\phi(\bar x,\bar y)$ and $\psi(\bar x,\bar z)$ are non-empty conjunctions of atoms that mention only variables from $\bar x \cup \bar y$ and $\bar x \cup \bar z$, respectively. Note that, by abuse of notation, we may treat a tuple of variables as a set of variables.
We write $\sigma$ as $\phi(\bar x,\bar y) \ra \exists \bar z\, \psi(\bar x,\bar z)$, and use comma instead of $\wedge$ for joining atoms. We refer to $\phi(\bar x,\bar y)$ and $\psi(\bar x,\bar z)$ as the {\em body} and {\em head} of $\sigma$, denoted $\body{\sigma}$ and $\head{\sigma}$, respectively.
The {\em frontier} of the TGD $\sigma$, denoted $\fr{\sigma}$, is the set of variables $\bar x$, i.e., the variables that appear both in the body and the head of $\sigma$. 
The {\em schema} of a set $\dep$ of TGDs, denoted $\sch{\dep}$, is the set of predicates occurring in $\dep$.
%
%
An instance $I$ satisfies a TGD $\sigma$ as the one above, written $I \models \sigma$, if whenever there exists a homomorphism $h$ from $\phi(\bar x, \bar y)$ to $I$, then there is $h' \supseteq h_{|\bar x}$ that is a homomorphism from $\psi(\bar x,\bar z)$ to $I$; we may treat a conjunction of atoms as a set of atoms. The instance $I$ satisfies a set $\dep$ of TGDs, written $I \models \dep$, if $I \models \sigma$ for each $\sigma \in \dep$.

\medskip

\noindent
\textbf{Linearity.} A TGD is called {\em linear} if it has only one body-atom, and the corresponding class that collects all the finite sets of linear TGDs is denoted $\class{L}$. We further call a linear TGD {\em simple} if no variable occurs more than once in its body-atom, and the corresponding class is denoted $\class{SL}$. It is clear that $\class{SL} \subsetneq \class{L}$.
\section{The Semi-Oblivious Chase Procedure}\label{sec:semi}

The semi-oblivious chase (or simply chase) takes as input a database $D$ and a set $\dep$ of TGDs, and constructs an instance that contains $D$ and satisfies $\dep$.
A central notion in this context is that of trigger.

\begin{definition}
	Given a set $\dep$ of TGDs and an instance $I$, a {\em trigger} for $\dep$  on $I$ is a pair $(\sigma,h)$, where $\sigma \in \dep$ and $h$ is a homomorphism from $\body{\sigma}$ to $I$.
	The {\em result} of $(\sigma,h)$, denoted $\result{\sigma}{h}$, is the set $\mu(\head{\sigma})$, where $\mu : \var{\head{\sigma}} \ra \ins{C} \cup \ins{N}$ is defined as follows:
	%
	\[
	\mu(x)\
	=\ \left\{
	\begin{array}{ll}
	h(x) & \quad \text{if } x \in \fr{\sigma}\\
	&\\
	\bot_{\sigma,h_{|\fr{\sigma}}}^{x} & \quad \text{otherwise}
	\end{array} \right.
	\]
	where $\bot_{\sigma,h_{|\fr{\sigma}}}^{x} \in \ins{N}$.  Let $T(\dep,I)$ be the set of triggers for $\dep$ on $I$.	\hfill\markfull
\end{definition}

Observe that in the definition of $\result{\sigma}{h}$, each existentially quantified variable $x$ of $\head{\sigma}$ is mapped by $\mu$ to a null value of $\ins{N}$ whose name is uniquely determined by the trigger $(\sigma,h)$ and the variable $x$ itself. This means that, given a trigger $(\sigma,h)$, we can unambiguously construct the set of atoms $\result{\sigma}{h}$.
The central idea of the chase is, starting from a database $D$, to exhaustively apply triggers for the given set $\dep$ of TGDs on the instance constructed so far.
More precisely, given a database $D$ and a set $\dep$ of TGDs, let
\[
\mathsf{chase}^{0}(D,\dep)\ =\ D,
\]
and for each $i>0$, let
\[
\mathsf{chase}^{i}(D,\dep)\ =\ \mathsf{chase}^{i-1}(D,\dep)\ \cup\ \bigcup_{(\sigma,h) \in S} \result{\sigma}{h},
\]
where $S = T(\dep,\mathsf{chase}^{i-1}(D,\dep))$. 
We finally define {\em the result of the chase of $D$ w.r.t.~$\dep$} as the (possibly infinite) instance
\[
\chase{D}{\dep}\ =\ \bigcup_{i \geq 0} \mathsf{chase}^{i}(D,\dep).
\]

\medskip

\noindent
\textbf{Chase Termination.}
The result of the chase may be infinite even for very simple settings: it is easy to see that for $D = \{R(a,b)\}$ and $\dep = \{R(x,y) \ra \exists z \, R(y,z)\}$, $\chase{D}{\dep}$ is infinite.
%
This leads to the following problem, parameterized by a class $\class{C}$ of TGDs such as $\class{SL}$ (the class of simple-linear TGDs) and $\class{L}$ (the class of linear TGDs):

\medskip

\begin{center}
	\fbox{
		\begin{tabular}{ll}
			{\small INPUT} : & A database $D$ and a set $\dep$ of TGDs from $\class{C}$.
			\\
			{\small QUESTION} : &  Is the instance $\chase{D}{\dep}$ finite?
	\end{tabular}}
\end{center}

\medskip

\noindent This problem has been recently studied in~\cite{CaGP22} for the classes of simple-linear and linear TGDs. Interestingly, for both classes, the finiteness of the result of the chase has been syntactically characterized by exploiting the notion of non-uniform weak-acyclicity. 
We proceed to recall this acyclicity notion, and then present the characterizations established in~\cite{CaGP22}, which in turn lead to simple algorithms for checking the finiteness of the result of the chase.
Note that, for the sake of clarity, in the rest of the paper we assume TGDs with a non-empty frontier, i.e., we assume that there is at least one variable in a TGD $\sigma$ that occurs both in $\body{\sigma}$ and $\head{\sigma}$. This assumption can be made without loss of generality since, given a database $D$ and a set $\dep$ of TGDs, we can easily construct a set $\dep'$ of TGDs with a non-empty frontier by slightly modifying $\dep$ such that $\chase{D}{\dep}$ is finite iff $\chase{D}{\dep'}$ is finite.

\medskip

\noindent
\textbf{Non-Uniform Weak-Acyclicity.} Weak-acyclicity was introduced in~\cite{FKMP05} as the main formalism for data exchange purposes, which guarantees the finiteness of the result of the chase for {\em every} input database. Non-uniform weak-acyclicity is the database-dependent variant of weak-acyclicity introduced in~\cite{CaGP22}. We proceed to give the formal definitions.
We first need to recall the notion of the {\em dependency graph} of a set $\dep$ of TGDs, 
defined as a directed multigraph $\depg{\dep}=(N,E)$, where $N = \pos{\sch{\dep}}$ and $E$ contains {\em only} the following edges.
For each TGD $\sigma \in \dep$ with $\head{\sigma} = \{\alpha_1,\ldots,\alpha_k\}$, for each $x \in \frontier{\sigma}$, and for each position $\pi \in \posvar{\body{\sigma}}{x}$:
\begin{itemize}
	\item For each $i \in [k]$ and for each $\pi' \in \posvar{\alpha_i}{x}$, there exists a \emph{normal} edge $(\pi,\pi') \in E$.
	\item For each existentially quantified variable $z$ in $\sigma$, $i \in [k]$, and $\pi' \in \posvar{\alpha_i}{z}$, there is a \emph{special} edge $(\pi,\pi') \in E$.
\end{itemize}
We further need to define when a predicate is reachable from another predicate. 
Given predicates $R,P \in \sch{\dep}$, {\em $P$ is reachable from $R$ (w.r.t.~$\dep$)} if $R = P$, or there exists a path in $\depg{\dep}$ from a position of the form $(R,i)$ to a position of the form $(P,j)$.
%
Given a database $D$, we say that a (not necessarily simple and possibly cyclic) path $C$ in $\depg{\dep}$ is \emph{$D$-supported} if there exists an atom $R(\bar t) \in D$ and a node of the form $(P,i)$ in $C$ such that $P$ is reachable from $R$.
We are now ready to recall (non-uniform) weak-acyclicity.

\begin{definition}\label{def:dwa}
	Consider a database $D$ and a set $\dep$ of TGDs. We say that $\dep$ is {\em weakly-acyclic w.r.t.~$D$}, or {\em $D$-weakly-acyclic}, if there is no $D$-supported cycle in $\depg{\dep}$ with a special edge. 
	We say that $\dep$ is {\em weakly-acyclic} if there is no cycle in $\depg{\dep}$ with a special edge. \hfill\markfull
\end{definition}

\smallskip

\noindent
\textbf{Characterizing the Finiteness of the Chase.}
It is not very difficult to show that whenever a set $\dep$ of TGDs (not necessarily linear) is $D$-weakly-acyclic, then the instance $\chase{D}{\dep}$ is finite. In other words, the $D$-weak-acyclicity of $\dep$ is a sufficient condition for the finiteness of $\chase{D}{\dep}$. What is more interesting is that, assuming that $\dep$ is a set of simple-linear TGDs, the $D$-weak-acyclicity of $\dep$ is also a necessary condition for the finiteness of $\chase{D}{\dep}$. This leads to the following characterization established in~\cite{CaGP22}:

\begin{theorem}\label{the:characterization-simple-linear}
	Consider a database $D$ and a set $\dep \in \class{SL}$ of TGDs. It holds that $\chase{D}{\dep}$ is finite iff $\dep$ is $D$-weakly-acyclic.
\end{theorem}

For linear TGDs, it turned out that non-uniform weak-acyclicity is not powerful enough for characterizing the finiteness of the chase instance. Here is an example given in~\cite{CaGP22} that illustrates this fact:

\begin{example}
	Consider the database $D = \{R(a,b)\}$ and the singleton set $\dep$ consisting of the (non-simple) linear TGD
	\[
	R(x,x)\ \ra\ \exists z \, R(z,x). 
	\]
	It is easy to see that there is no trigger for $\dep$ on $D$. This means that $\chase{D}{\dep} = D$ is finite, whereas $\dep$ is {\em not} $D$-weakly-acyclic. \hfill\markfull
\end{example}

To obtain a characterization analogous to Theorem~\ref{the:characterization-simple-linear}, the authors of~\cite{CaGP22} used the technique of {\em simplification} to convert linear TGDs into simple-linear TGDs, while preserving the finiteness of the chase instance. We proceed to recall this technique.
Let $\bar t = (t_1,\ldots,t_n)$ be a tuple of (not necessarily distinct) terms. We write $\unique{\bar t}$ for the tuple obtained from $\bar t$ by keeping only the first occurrence of each term in $\bar t$.
For example, if $\bar t = (x,y,x,z,y)$, then $\unique{\bar t} = (x,y,z)$.
For each $i \in [n]$, the \emph{identifier of $t_i$ in $\bar t$}, denoted $\id{\bar t}{t_i}$, is the integer that identifies the position of $\unique{\bar t}$ at which $t_i$ appears. 
We write $\id{}{\bar t}$ for the tuple $(\id{\bar t}{t_1},\ldots,\id{\bar t}{t_n})$.
For example, if $\bar t = (x,y,x,z,y)$, then $\id{}{\bar t} = (1,2,1,3,2)$.
For an atom $\alpha = R(\bar t)$, the {\em simplification of $\alpha$}, denoted $\simple{\alpha}$, is the atom $R_{\id{}{\bar t}}(\unique{\bar t})$, whereas the {\em shape of $\alpha$}, denoted $\shape{\alpha}$, is the predicate $R_{\id{}{\bar t}}$. We can naturally refer to the simplification and the shape of a set of atoms.
For a tuple of variables $\bar x = (x_1,\ldots,x_n)$, a \emph{specialization of $\bar x$} is a function $f$ from $\bar x$ to $\bar x$ such that $f(x_1) = x_1$, and $f(x_i) \in \{f(x_1),\ldots,f(x_{i-1}),x_i\}$, for each $i \in \{2,\ldots,n\}$.
We write $f(\bar x)$ for $(f(x_1),\ldots,f(x_n))$. We are now ready to recall how a set of linear TGDs is converted into a set of simple-linear TGDs.

\begin{definition}\label{def:simplification}
	Consider a linear TGD $\sigma$ of the form
	\[
	R(\bar x) \ra \exists \bar z\, \psi(\bar y,\bar z), 
	\]
	where $\bar y \subseteq \bar x$, and a specialization $f$ of $\bar x$. The {\em simplification of $\sigma$ induced by $f$} is the simple-linear TGD
	\[
	\simple{R(f(\bar x))} \rightarrow \exists \bar z\, \simple{\psi(f(\bar y),\bar z)}.
	\]
	We write $\simple{\sigma}$ for the set of all simplifications of $\sigma$ induced by some specialization of $\bar x$.
	For a set $\dep \in \class{L}$ of TGDs, the {\em simplification of $\dep$} is defined as the set
	\[
	\simple{\dep}\ =\ \bigcup_{\sigma \in \dep} \simple{\sigma}
	\]
	consisting only of simple-linear TGDs. \hfill\markfull
\end{definition}

We can now recall the characterization for the finiteness of the chase instance for linear TGDs, established in~\cite{CaGP22}, which is similar to the one for simple-linear TGDs, with the key difference that first we need to simplify both the database and the set of linear TGDs:

\begin{theorem}\label{the:characterization-linear}
	Consider a database $D$ and a set $\dep \in \class{L}$ of TGDs. Then, $\chase{D}{\dep}$ is finite iff $\simple{\dep}$ is $\simple{D}$-weakly-acyclic.
\end{theorem}

It is clear that Theorems~\ref{the:characterization-simple-linear} and~\ref{the:characterization-linear} provide simple algorithms for checking whether the chase instance is finite. In particular, given a database $D$ and a set $\dep$ of simple-linear TGDs, we simply need to check whether $\dep$ is $D$-weakly-acyclic, in which case the algorithm returns \true; otherwise, it returns \false. The same holds when $\dep$ is a set of linear TGDs, with the difference that the algorithm first needs to simplify $D$ and $\dep$, and then perform the acyclicity check.
Our goal is to experimentally evaluate the above algorithms with the aim of understanding which input parameters affect their performance, clarifying whether they can be applied in a practical context, and revealing their performance limitations. Of course, a naive implementation of the above algorithms, especially for linear TGDs where the expensive simplification must be applied, will lead to poor performance, and thus, will not be very useful towards our goal. Hence, we need to somehow convert the above theoretical algorithms into practical algorithms that are amenable to efficient implementations. This is the subject of the next section.
\section{Practical Termination Algorithms} \label{sec:alg}

We first present the algorithm $\mathsf{IsChaseFinite[SL]}$ that accepts as input a database $D$ and a set $\dep$ of simple-linear TGDs, and checks whether $\dep$ is $D$-weakly-acyclic, which is equivalent to say that the instance $\chase{D}{\dep}$ is finite.
Note that a naive search for a ``bad'' cycle in a dependency graph will be too costly since we may have to go through exponentially many cycles. Thus, $\mathsf{IsChaseFinite[SL]}$ relies on a refined machinery that searches for {\em strongly connected components} with a special edge.
We then proceed to give an analogous algorithm, dubbed $\mathsf{IsChaseFinite[L]}$, for linear TGDs, which essentially simplifies the given database $D$ and set $\dep$ of linear TGDs, and then checks whether $\simple{\dep}$ is $\simple{D}$-weakly-acyclic, which is equivalent to say that $\chase{D}{\dep}$ is finite.
%
Note, however, that $\mathsf{IsChaseFinite[L]}$ relies on a refined notion of simplification that {\em dynamically simplifies} $\dep$ by leveraging the given database $D$, instead of doing it statically as in Definition~\ref{def:simplification} without taking any database into account. The goal of the dynamic simplification is to keep only TGDs of $\simple{\dep}$ that are really needed for checking whether the chase instance is finite.
Note that in this section we present the above algorithms at a high-level without delving into implementation details; the latter will be the subject of Section~\ref{sec:implementation}.

\subsection{Simple-Linear TGDs}\label{sec:slinear}

Before presenting $\mathsf{IsChaseFinite[SL]}$, we need to introduce a couple of auxiliary notions. A {\em strongly connected component} (SCC) in a directed graph $G$ is a maximal subgraph of $G$ in which there is a (directed) path between every pair of nodes. A {\em special SCC} in a dependency graph is an SCC with at least one special edge. We are now ready to discuss $\mathsf{IsChaseFinite[SL]}$, which is depicted in Algorithm~\ref{alg:slinear}.
It starts by building the dependency graph $G$ of the input set $\dep$ of simple-linear TGDs (line~\ref{ln:graph}).
It then collects the special SCCs of $G$ in a set $S$ (line~\ref{ln:bcycles}), which can clearly form ``bad'' cycles that violate the condition underlying non-uniform weak-acyclicity. Of course, for the latter to happen, some nodes (i.e,., predicate positions) in a special SCC must be supported by the given database $D$ as defined in Section~\ref{sec:prel}. To check this, the algorithm first collects exactly one node $v_C$ from each special SCC $C$ of $G$ in a set $P$ (line~\ref{ln:initP}); note that it is not important how $v_C$ is selected. It then checks if $D$ supports any of the nodes of $P$ (line~\ref{ln:return_f}). If this is the case, then there is a $D$-supported cycle in $G$ with a special edge, and thus the algorithm returns \false; otherwise, it returns \true. The implementation details of \textsf{BuildDepGraph}, \textsf{FindSpecialSCC}, and \textsf{Supports} are discussed in Sections~\ref{sec:bgarph}, \ref{sec:scc}, and \ref{sec:support}, respectively.
The correctness of $\mathsf{IsChaseFinite[SL]}$ follows by Theorem~\ref{the:characterization-simple-linear}:

\begin{lemma}\label{lm:simple}
	Consider a database $D$ and a set $\dep \in \class{SL}$ of TGDs. It holds that $\mathsf{IsChaseFinite}[\class{SL}](D,\dep) = \true$ iff $\chase{D}{\dep}$ is finite.
\end{lemma}

\begin{algorithm}[t]
\KwIn{A database $D$ and a set $\dep \in \class{SL}$ of TGDs}
\KwOut{\true if $\chase{D}{\dep}$ is finite and \false otherwise}

\medskip

$G \leftarrow \textsf{BuildDepGraph}(\dep)$;\label{ln:graph}\\
$S \leftarrow \textsf{FindSpecialSCC}(G)$;\label{ln:bcycles}\\
$P \leftarrow \bigcup_{C \in S}{\{v_C\}}$;\label{ln:initP}\\
\lIf{$\mathsf{Supports}(D,P,G)$}{\KwRet{\false}\label{ln:return_f}}
\KwRet{\true}\label{ln:return_t}
\caption{$\mathsf{IsChaseFinite[SL]}$}\label{alg:slinear}
\end{algorithm}

\subsection{Linear TGDs}\label{sec:linear}

Although the algorithm $\mathsf{IsChaseFinite[SL]}$ together with the simplification technique (see Definition~\ref{def:simplification}) immediately give rise to a simple algorithm for checking the finiteness of the chase instance for linear TGDs, a naive implementation of the simplification technique leads to poor performance. Indeed, we performed exploratory experiments on real-world sets of linear TGDs and observed that a naive implementation is not scalable as the algorithm quickly runs out of memory when dealing with large sets of TGDs. This is because by statically simplifying a set of linear TGDs $\dep$, without taking into account the underlying database, leads to an exponentially large set of simple-linear TGDs; in particular, the size of the set $\simple{\dep}$ is exponential in the maximum arity of the predicates in $\sch{\dep}$. Thus, the algorithm $\mathsf{IsChaseFinite[SL]}$ becomes impractical due to the very large size of the dependency graph of $\simple{\dep}$, which exceeds the capacity of the main memory.

\medskip

\noindent
\textbf{Dynamic Simplification.}
We refine the notion of simplification by taking into account the underlying database, which leads to the technique of dynamic simplification. In particular, given a database $D$ and a set $\dep$ of linear TGDs, the goal is to define a set $\dsimple{D}{\dep}$, which is a subset of $\simple{\dep}$, that enjoys two crucial properties:
\begin{enumerate}
	\item It holds that the instance $\chase{\simple{D}}{\simple{\dep}}$ is finite iff the instance $\chase{\simple{D}}{\dsimple{D}{\dep}}$ is finite, which essentially tells us that the technique of dynamic simplification preserves the finiteness of the chase.
	\item The set $\dsimple{D}{\dep}$ is, in general, orders of magnitude smaller than the set $\simple{\dep}$ obtained by statically simplifying $\dep$.
\end{enumerate}
Item (1) is established by Lemma~\ref{lm:dyn-simplification} below. Item (2) cannot be mathematically proved as there are cases where both static and dynamic simplification build the same set of linear TGDs. However, we have experimentally verified that for existing databases and sets of TGDs coming from the literature (in fact, those used in Section~\ref{sec:rw-kb}), the size of the dynamically simplified sets of TGDs is, on average, 5 times smaller than the size of the corresponding statically simplified sets of TGDs. The absolute difference varies with the dynamically simplified sets being up to 1000 times smaller in the best case.

The key idea of dynamic simplification is to exploit the shapes of the atoms occurring in the given database to guide the simplification.
More precisely, given a database $D$ and a set $\dep$ of linear TGDs, we first collect the shapes that can be derived from $\shape{D}$ using the TGDs of $\dep$; we denote this set as $\dep(\shape{D})$. Then, $\dsimple{D}{\dep}$ keeps from the set $\simple{\dep}$ only those simple-linear TGDs such that the predicate of their body-atom belongs to $\dep(\shape{D})$, as these are the only TGDs that can be applied during the construction of the instance $\chase{\simple{D}}{\simple{\dep}}$. All the other TGDs of $\simple{\dep}$ are superfluous whenever the input database is $D$ in the sense that they will never be applied during the construction of $\chase{\simple{D}}{\simple{\dep}}$. We proceed to formalize this idea. To this end, we need to introduce some auxiliary notions.

For a schema $\ins{S}$, let $\shape{\ins{S}}$ be the set of all shapes mentioning a predicate of $\ins{S}$, that is, the finite set of shapes
\[
\shape{\ins{S}}\ =\ \left\{R_{\id{}{\bar t}} \mid R \in \ins{S} \textrm{ and } \bar t \in \left(\ins{C}^{\arity{R}} \cup \ins{V}^{\arity{R}}\right)\right\}.
\]
For a set of shapes $S \subseteq \shape{\ins{S}}$, the {\em database induced by $S$}, denoted $\mi{DB}[S]$, is the database $\{R(\id{}{\bar t}) \mid R_{\id{}{\bar t}} \in S\}$. For example, assuming that $S = \{R_{(1,2)}, P_{(1,1,2)}\}$, then 
$\mi{DB}[S] = \{R(1,2),P(1,1,2)\}$.
Consider now a linear TGD $\sigma = 
R(x_1,\ldots,x_n) \ra \exists \bar z\, \psi(\bar y,\bar z)$
and let $h$ be a homomorphism from $\{R(x_1,\ldots,x_n)\}$ to $\{R(i_1,\ldots,i_n)\} \subseteq \mi{DB}[\shape{\{R\}}]$. The {\em $h$-specialization} of the tuple $(x_1,\ldots,x_n)$ is the (unique) specialization $f$ of $(x_1,\ldots,x_n)$ such that $f(x_i) = f(x_j)$ iff $h(x_i) = h(x_j)$, for every $i,j \in [n]$. For example, assuming that $h$ is a homomorphism from $\{R(x,y,x,z)\}$ to $\{R(1,1,1,2)\}$, the $h$-specialization of $(x,y,x,z)$ is the function $f$ such that $f(x)=x$, $f(y)=x$, and $f(z)=z$.
We can now proceed with the formalization of dynamic simplification.

Consider a set $\dep$ of linear TGDs and a set of shapes $S \subseteq \shape{\dep}$; for brevity, we write $\shape{\dep}$ for $\shape{\sch{\dep}}$. A shape $R_{\id{}{\bar t}} \in \shape{\dep}$ is an {\em immediate consequence} of $S$ and $\dep$ if:
\begin{enumerate}
	\item $R_{\id{}{\bar t}} \in S$, or
	\item there is a TGD $R(\bar x) \ra \exists \bar z\, \psi(\bar y,\bar z)$ in $\dep$ and a homomorphism $h$ from $\{R(\bar x)\}$ to $\mi{DB}[S]$ such that $R_{\id{}{\bar t}}$ occurs in the head of the simplification of $\sigma$ induced by the $h$-specialization of $\bar x$.
\end{enumerate}
In simple words, item (2) tells us that there exists a TGD in $\simple{\dep}$ of the form $R'_{\id{}{\bar t'}}(\bar x)\ \ra\ \exists \bar z \, \ldots,R_{\id{}{\bar t}}(\bar y),\ldots$ with $R'_{\id{}{\bar t'}} \in S$.
The {\em immediate consequence operator} of $\dep$ is the function $\Gamma_{\dep} : 2^{\shape{\dep}} \ra 2^{\shape{\dep}}$
(as usual, $2^X$ denotes the powerset of a set $X$) such that
\[
\Gamma_{\dep}(S)\ =\ \left\{R_{\id{}{\bar t}} \mid R_{\id{}{\bar t}} \text{ is an immediate consequence of } S \text{ and } \dep\right\}.
\]
By iterative applications of the above operator, we can compute the shapes that can be derived from $S$ using the TGDs of $\dep$. Formally,
\[
\Gamma_{\dep}^{0}(S)\ =\ S \qquad \text{and} \qquad \Gamma_{\dep}^{i}(S)\ =\ \Gamma_{\dep}(\Gamma_{\dep}^{i-1}(S)) \,\, \text{ for } \,\, i > 0
\]
and we finally let
\[
\dep(S)\ =\ \bigcup_{i \geq 0} \Gamma_{\dep}^{i}(S).
\]
At first glance, the construction of $\dep(S)$ requires infinitely many iterations. However, since $\dep(S) \subseteq \shape{\dep}$, in the worst-case $\dep(S)$ is obtained after $|\shape{\dep}|$ iterations. It is actually easy to verify that $\dep(S) = \Gamma_{\dep}^{|\simple{\dep}|}(S)$. Therefore, since $\shape{\dep}$ is finite, we conclude that $\dep(S)$ can be obtained after finitely many steps.
We now have all the ingredients to formally define dynamic simplification.

\begin{definition}\label{def:dyn-simplification}
	Consider a database $D$ and a set $\dep$ of linear TGDs.\footnote{We assume, without loss of generality, that the atoms of $D$ mention only predicates of $\sch{\dep}$, and thus, $\shape{D} \subseteq \shape{\dep}$. Indeed, the atoms of $D$ with a predicate not in $\sch{\dep}$ do not affect in any way the size of the instance $\chase{D}{\dep}$.} The {\em dynamic simplification of $\dep$ relative to $D$} (or {\em $D$-simplification of $\dep$}), denoted $\dsimple{D}{\dep}$, is defined as the set
	\begin{multline*}
	\big\{	\simple{R(f(\bar x))} \rightarrow \exists \bar z\, \simple{\psi(f(\bar y),\bar z)} \mid\\
	R(\bar x) \ra \exists \bar z\, \psi(\bar y,\bar z) \in \dep \text{ and } f \text{ is the $h$-specialization of } \bar x\\
	\text{ for some homomorphism } h \text{ from } \{R(\bar x)\} \text{ to } \mi{DB}(\dep(\shape{D}))\big\}
	\end{multline*}
	consisting only of simple-linear TGDs. \hfill\markfull
\end{definition}

It is not difficult to verify that the $D$-simplification of $\dep$ essentially collects all the TGDs of $\simple{\dep}$ such that the predicate of their body-atom belongs to $\dep(\shape{D})$.
\ignore{
\begin{definition}\label{def:dyn-simplification}
	Consider a database $D$ and a set $\dep$ of linear TGDs.\footnote{We assume, without loss of generality, that the atoms of $D$ mention only predicates of $\sch{\dep}$, and thus, $\shape{D} \subseteq \shape{\dep}$. Indeed, the atoms of $D$ with a predicate not in $\sch{\dep}$ do not affect in any way the size of the instance $\chase{D}{\dep}$.} The {\em dynamic simplification of $\dep$ relative to $D$} (or {\em $D$-simplification of $\dep$}), denoted $\dsimple{D}{\dep}$, is defined as the set
	\[
	\left\{R_{\id{}{\bar t}}(\bar x) \ra \exists \bar z \, \psi(\bar y,\bar z) \in \simple{\dep} \mid R_{\id{}{\bar t}} \in \dep(\shape{D})\right\}.
	\]
	consisting only of simple-linear TGDs. \hfill\markfull
\end{definition}
}
We now proceed to show that indeed dynamic simplification preserves the finiteness of the chase.

\begin{lemma}\label{lm:dyn-simplification}
	Consider a database $D$ and a set $\dep \in \class{L}$ of TGDs. The following are equivalent:
	\begin{enumerate}
		\item $\chase{\simple{D}}{\simple{\dep}}$ is finite.
		\item $\chase{\simple{D}}{\dsimple{D}{\dep}}$ is finite.
	\end{enumerate}
\end{lemma}

\begin{proof}
	Since, by definition, $\dsimple{D}{\dep} \subseteq \simple{\dep}$, it is clear that (1) implies (2). The interesting direction is (2) implies (1).
	\ignore{
	We start by observing that $(1)$ holds iff $\chase{\simple{D}}{\simple{\dep^+}}$ is finite, and $(2)$ holds iff $\chase{\simple{D}}{\dsimple{D}{\dep^+}}$ is finite, where $\dep^+$ is obtained from $\dep$ by replacing each TGD with an empty frontier of the form
	$R(x_1,\ldots,x_n) \ra \exists z \, \psi(\bar z)$ with the TGD $R(x_1,\ldots,x_n,y) \ra \exists z \, \psi(\bar z,y)$ with a non-empty frontier; $\psi(\bar z,y)$ is obtained from $\psi(\bar z)$ by simply replacing each atom $P(\bar u)$ with $P(\bar u,y)$.
	Therefore, it suffices to show that if $\chase{\simple{D}}{\dsimple{D}{\dep^+}}$ is finite, then $\chase{\simple{D}}{\simple{\dep^+}}$ is finite. As we shall see, it is more convenient to work with $\dep^+$ instead of $\dep$ since the following is guaranteed: whenever a predicate $P$ is reachable from a predicate $R$ (w.r.t.~$\dep^+$), i.e., $R \reach{\dep^+} P$, then there exists a path in the dependency graph of $\dep^+$ starting from a node $(R,i)$ and ending at a node $(P,j)$; the latter is not true whenever we have TGDs with an empty frontier. 
}
	We proceed to establish the contrapositive of the implication in question, that is, if the instance $\chase{\simple{D}}{\simple{\dep}}$ is infinite, then the instance $\chase{\simple{D}}{\dsimple{D}{\dep}}$ is also infinite. To this end, we first show an auxiliary technical lemma:
	
	\begin{lemma}\label{lm:aux-dyn-simplification}
		Consider a path $(R_1,i_1),\ldots,(R_n,i_n)$ in the dependency graph of $\simple{\dep}$ such that there is an atom of the form $R_1(\bar c)$ in $\simple{D}$. It holds that $\{R_1,\ldots,R_n\} \subseteq \dep(\shape{D})$.
	\end{lemma}

	\begin{proof}
		We proceed by induction on the length $n$ of the path.
				
		\textbf{Base Case.} Clearly, $\shape{D} \subseteq \dep(\shape{D})$. Since $\shape{D}$ coincides with the set of predicates used by the atoms of $\simple{D}$ and, by hypothesis, $R_1$ is used by an atom of $\simple{D}$, we get that $R_1 \in \dep(\shape{D})$, as needed.
				
		\textbf{Inductive Step.} Consider a path $(R_1,i_1),\ldots,(R_n,i_n)$ in the dependency graph of $\simple{\dep}$ with $R_1$ being a predicate used by an atom of $\simple{D}$. By induction hypothesis, $\{R_1,\ldots,R_{n-1}\} \subseteq \dep(\shape{D})$. It remains to show that $R_n \in \dep(\shape{D})$.
		Assume that $\sigma$ is the TGD witnessing the edge $((R_{n-1},i_{n-1}),(R_{n},i_{n}))$ in the dependency graph of $\simple{\dep}$. There is $i \geq 0$ such that $R_{n-1} \in \Gamma_{\dep}^{i}(\shape{D})$. Since $R_{n-1}$ is the predicate of the body-atom of $\sigma$, $R_n \in \Gamma_{\dep}^{i+1}(\shape{D})$. Thus, $R_n \in \dep(\shape{D})$.
	\end{proof}

	We can now show the desired implication. Since, by hypothesis, $\chase{\simple{D}}{\simple{\dep}}$ is infinite, Theorem~\ref{the:characterization-simple-linear} allows us to conclude that there exists a $\simple{D}$-supported cycle with a special edge in the dependency graph of $\simple{\dep}$. Let
	\[
	(R_1,i_1),\ldots,(R_n,i_n),
	\]
	with $(R_1,i_1) = (R_n,i_n)$, be such a cycle. Since this cycle is $\simple{D}$-supported, there exists an atom $P(\bar c) \in D$ and a node $(R_j,i_j)$ in the cycle such that $P \reach{\dep} R_j$. Since the TGDs of $\dep$ have a non-empty frontier, in the dependency graph of $\simple{\dep}$ there exists a path 
	\[
	(P_1,k_1),\ldots,(P_m,k_m)
	\]
	with $P_1 = P$ and $(P_m,k_m) = (R_j,i_j)$. Summing up, in the dependency graph of $\simple{\dep}$, there exists a path of the form
	\begin{align*}
	& (P_1,k_1),\ldots,(P_{m-1},k_{m-1}),\\
	&\hspace{10mm}(R_j,i_j),\ldots,(R_{n-1},i_{n-1}),
	(R_1,i_1),\ldots,(R_{j-1},i_{j-1}),(R_j,i_j).
	\end{align*}
	Assume that this path is witnessed by the TGDs $\sigma_1,\ldots,\sigma_{n+m-2}$. By Lemma~\ref{lm:aux-dyn-simplification}, $\{P_1,\ldots,P_{m-1},R_1,\ldots,R_{n-1}\} \subseteq \dep(\shape{D})$. By the definition of dynamic simplification (see Definition~\ref{def:dyn-simplification}), we get that the TGDs $\sigma_1,\ldots,\sigma_{n+m-2}$ belong to $\dsimple{D}{\dep}$. Therefore, the above $\simple{D}$-supported cycle with a special edge also occurs in the dependency graph of $\dsimple{D}{\dep}$. Hence, by Theorem~\ref{the:characterization-simple-linear}, $\chase{\simple{D}}{\dsimple{D}{\dep}}$ is infinite, and the claim follows.
\end{proof}

Another crucial property of dynamic simplification, which will help us to further improve the performance of the termination algorithm for linear TGDs, is that the $\simple{D}$-weak-acyclicity of $\dsimple{D}{\dep}$ coincides with the weak-acyclicity of $\dsimple{D}{\dep}$. This holds since, by construction, all the predicates occurring in the TGDs of $\dsimple{D}{\dep}$ are reachable from a predicate occurring in $\simple{D}$, and thus, each cycle with a special edge in the dependency graph of $\dsimple{D}{\dep}$ is trivially $\simple{D}$-supported.
The above property immediately implies Lemma~\ref{lm:avoid-simplification} below since checking for the finiteness of $\chase{\simple{D}}{\dsimple{D}{\dep}}$, which is equivalent to the $\simple{D}$-weak-acyclicity of $\dsimple{D}{\dep}$ by Theorem~\ref{the:characterization-linear}, boils down to checking if $\dsimple{D}{\dep}$ is weakly-acyclic, without having to explicitly check for $\simple{D}$-supportedness, and thus, avoiding the expensive task of simplifying $D$.

\begin{lemma}\label{lm:avoid-simplification}
	Consider a database $D$ and a set $\dep \in \class{L}$ of TGDs. The following are equivalent:
	\begin{enumerate}
	\item $\chase{\simple{D}}{\dsimple{D}{\dep}}$ is finite.
	\item $\dsimple{D}{\dep}$ is weakly-acyclic.
	\end{enumerate}
\end{lemma}

\begin{algorithm}[t]
	\KwIn{A database $D$ and a set $\dep \in \class{L}$ of TGDs}
	\KwOut{The $D$-simplification of $\dep$}
	
	\medskip
	
	$S \leftarrow \mathsf{FindShapes}(D)$;\label{ln:shapes}\\
	$\Sigma_s \leftarrow \emptyset$;\label{ln:empty}\\
	$\Delta S \leftarrow S$;\label{ln:delta}\\
	\While{$\Delta S\neq \emptyset$\label{ln:while}}{
		$\dep_{\mi{aux}} \leftarrow \mathsf{Applicable}(\Delta S,\Sigma)$;\label{ln:tgdsaux}\\
		$S_{\mi{aux}} \leftarrow \big\{R_{\id{}{\bar t}} \in \shape{\dep} \mid \text{ there exists a TGD } \sigma \in \dep_{\mi{aux}} \text{ such that } R_{\id{}{\bar t}} \text{ occurs in } \head{\sigma}\big\}$;\label{ln:shapesaux}\\
		$\Sigma_s \leftarrow \Sigma_s \cup \dep_{\mi{aux}}$;\label{ln:tgds}\\
		$\Delta S \leftarrow S_{\mi{aux}} \setminus S$;\label{ln:deltaS}\\
		$S \leftarrow S\cup \Delta S$\label{ln:S};
	}
	\KwRet{$\Sigma_s$;}\label{ln:returnD}
	\caption{$\mathsf{DynSimplification}$}\label{alg:dsimplify}
\end{algorithm}

\ignore{
\begin{algorithm}[t]
	\KwIn{A database $D$ and a set $\dep \in \class{L}$ of TGDs}
	\KwOut{The $D$-simplification of $\dep$}
	
	\medskip
	
	$S \leftarrow \mathsf{FindShapes}(D)$;\label{ln:shapes}\\
	$\Delta S \leftarrow S$;\label{ln:delta}\\
	$\Sigma' \leftarrow \emptyset$;\label{ln:empty}\\
	\While{$\Delta S\neq \emptyset$\label{ln:while}}{
		$\Sigma'\leftarrow \Sigma'\cup \mathsf{Applicable}(\Delta S,\Sigma)$;\label{ln:tgds}\\
		$\Delta S \leftarrow \Gamma_\Sigma(\Delta S) \setminus S$;\label{ln:deltaS}\\
		$S \leftarrow S\cup \Delta S$\label{ln:S};
	}
	\KwRet{$\Sigma'$;}\label{ln:returnD}
	\caption{$\mathsf{DynSimplification}$}\label{alg:dsimplify}
\end{algorithm}
}

\medskip
\noindent \textbf{An Algorithm for Dynamic Simplification.}
Dynamic simplification indeed allow us to filter out from the set obtained by statically simplifying a set of linear TGDs superfluous simple-linear TGDs by exploiting the given database. It remains, however, to provide a concrete algorithm that performs the dynamic simplification of a set of linear TGDs that is amenable to an efficient implementation. To this end, we proceed to present the algorithm $\mathsf{DynSimplification}$, depicted in Algorithm~\ref{alg:dsimplify}, that takes as input a database $D$ and a set $\dep$ of linear TGDs, and constructs the $D$-simplification of $\dep$.
The algorithm starts by finding the shapes of the atoms occurring in $D$, namely it computes the set $\shape{D}$ (line~\ref{ln:shapes}). It then initializes the set of simplified TGDs $\Sigma_s$ (line~\ref{ln:empty}) and the set of new shapes $\Delta S$ (line~\ref{ln:delta}).
Then, the algorithm iteratively generates simplified TGDs and collects the new shapes that are added to $\Delta S$,
and continues this until a fixpoint is reached, i.e., $\Delta S =\emptyset$ (line~\ref{ln:while}). 
In particular, at each iteration, the algorithm computes simplified TGDs that are not superfluous, i.e., they can be applied during the construction of $\chase{\simple{D}}{\simple{\dep}}$, that are added to $\Sigma_s$ (lines~\ref{ln:tgdsaux} and~\ref{ln:tgds}). This is done via the procedure $\mathsf{Applicable}$, which takes as input a set of shapes $\hat{S}$ and a set of linear TGDs $\hat{\dep}$, and returns the set
\begin{multline*}
\big\{	\simple{R(f(\bar x))} \rightarrow \exists \bar z\, \simple{\psi(f(\bar y),\bar z)} \mid\\
R(\bar x) \ra \exists \bar z\, \psi(\bar y,\bar z) \in \hat{\dep} \text{ and } f \text{ is the $h$-specialization of } \bar x\\
\text{ for some homomorphism } h \text{ from } \{R(\bar x)\} \text{ to } \mi{DB}[\hat{S}]\big\}.
\end{multline*}
In essence, the procedure $\mathsf{Applicable}$ computes the set of TGDs of $\simple{\hat{\dep}}$ such that the predicate of their body belongs to $\hat{S}$.
The algorithm also collects the newly generated shapes, that is, the predicates occurring in the head of the TGDs of $\mathsf{Applicable}(\Delta S,\dep)$, that are added to $\Delta S$ (lines~\ref{ln:shapesaux} and~\ref{ln:deltaS}).
\ignore{
While finding the new shapes, the algorithm also computes the simplified TGDs that are not superfluous, i.e., they can be applied during the construction of $\chase{\simple{D}}{\simple{\dep}}$, and are added to $\Sigma'$ (line~\ref{ln:tgds}). This is done via the procedure $\mathsf{Applicable}$, which takes as input a set of shapes $\hat{S}$ and a set of linear TGDs $\hat{\dep}$, and returns the set
\[
\left\{R_{\id{}{\bar t}}(\bar x) \ra \exists \bar z \, \psi(\bar y,\bar z) \in \simple{\hat{\dep}} \mid R_{\id{}{\bar t}} \in \hat{S}\right\}.
\]
}
Note that at each iteration, the algorithm applies the TGDs on $\Delta S$, not on $S$, with the exception of the first iteration where $S = \Delta S$. This works because there are no new applicable TGDs on $S$ after the first iteration since the TGDs are linear (only one body-atom) and all the applicable TGDs on $S$ are applied during the first iteration.
The implementation details of \textsf{FindShapes} and \textsf{Applicable} are discussed in Sections~\ref{sec:bgarph} and \ref{sec:scc}, respectively.
$\mathsf{DynSimplification}$ is correct by construction:

\begin{lemma}\label{lm:dyn-simplification-algorithm}
	Consider a database $D$ and a set $\dep \in \class{L}$ of TGDs. It holds that $\mathsf{DynSimplification}(D,\dep) = \dsimple{D}{\dep}$.
\end{lemma}

\begin{algorithm}[t]
	\KwIn{A database $D$ and a set $\dep \in \class{L}$ of TGDs}
	\KwOut{\true if $\chase{D}{\dep}$ is finite and \false otherwise}
	
	\medskip
	
	$\dep_s \leftarrow \mathsf{DynSimplification}(D,\Sigma)$;\label{ln:d-simplify}\\
	$G \leftarrow \textsf{BuildDepGraph}(\dep_s)$;\\	\lIf{$\mathsf{FindSpecialSCC}(G) \neq \emptyset$}{\KwRet{\false}}
	\KwRet{\true}
	\caption{$\mathsf{IsChaseFinite[L]}$}\label{alg:dterm}
\end{algorithm}

\medskip

\noindent
\textbf{Termination Algorithm.} Having in place 
$\mathsf{DynSimplification}$, it is now straightforward to devise the algorithm $\mathsf{IsChaseFinite[L]}$, depicted in Algorithm~\ref{alg:dterm}, that checks for the finiteness of the chase in the case of linear TGDs.
The correctness of $\mathsf{DynSimplification}$, Theorem~\ref{the:characterization-linear}, Lemma~\ref{lm:dyn-simplification}, and Lemma~\ref{lm:avoid-simplification}, imply the correctness of the algorithm $\mathsf{IsChaseFinite[L]}$, and the next lemma follows:

\begin{lemma}\label{lm:linear} Given a database $D$ and a set $\dep \in \class{L}$ of TGDs, it holds that $\mathsf{IsChaseFinite[L]}(D,\dep) = \true$ iff $\chase{D}{\dep}$ is finite.
\end{lemma}


\section{Implementation Details}\label{sec:implementation}

We proceed to discuss the implementation details of the algorithms for checking the finiteness of the chase instance presented in Section~\ref{sec:alg}. In particular, we discuss the implementation choices that help to improve the performance of the algorithms, but are missing from the descriptions given in Section~\ref{sec:alg}.
In Sections~\ref{sec:bgarph}, \ref{sec:scc}, and \ref{sec:support} we discuss the procedures $\mathsf{BuildDepGraph}$, $\mathsf{FindSpecialSCC}$, and $\mathsf{Supports}$, respectively, used in Algorithm~\ref{alg:slinear}. The details of $\mathsf{DynSimplification}$ used in Algorithm~\ref{alg:dterm} are discussed in  Section~\ref{sec:d-simplification}.

\subsection{Build Dependency Graphs}\label{sec:bgarph}

The procedure $\mathsf{BuildDepGraph}$ takes as input a set $\dep$ of TGDs, and returns the dependency graph of $\dep$ in the form of an adjacency list. Recall that an adjacency list for a directed graph is a list of lists. Each list corresponds to a node $v$ of the graph, and the members in such a list represent the outgoing edges of $v$. Towards an implementation of such an adjacency list, we store a dependency graph as a list of {\em node objects}. Each node object represents a node in the graph, and has a list of {\em edge objects} representing the edges of the graph. For each edge object, we additionally store a binary value that specifies whether the corresponding edge is special or not. Although a singly linked list suffices for storing a dependency graph, we implement a dependency graph's adjacency list as a doubly linked list for performance purposes. By implementing a doubly linked list, each node object, in addition to a list of edge objects, has a list of reverse edge objects representing the edges in the opposite direction. These reverse edge objects enable traversing the graph in the opposite direction of the edges, which significantly helps when checking the support of positions in special SCCs, as we explain in Section~\ref{sec:support}. 
Now, given a set $\dep$ of simple-linear TGDs, $\mathsf{BuildDepGraph}$ iterates over all TGDs and constructs the dependency graph of $\dep$ by creating new elements for newly visited positions and linking them as dictated by the TGDs. To speed up the process of building dependency graphs, the procedure also uses an index structure that maps predicate positions to their corresponding elements in the adjacency list. This index allows for fast access to the elements (nodes) for adding new links (edges) while parsing new TGDs. Using the index structure, the procedure creates dependency graphs in linear time w.r.t. the size of the input set of TGDs.

\subsection{Find Special SCCs}\label{sec:scc}

To implement $\mathsf{FindSpecialSCC}$ we adapt the well-known {\em Tarjan's algorithm} for finding SCCs in directed graphs~\cite{tarjan1972depth}, which we briefly recall below.
Other algorithms for finding SCCs can be found in the literature (see, e.g., \cite{dijkstra1982finding,sharir1981strong}), some of which are simpler than Tarjan's algorithm such as Kosaraju's algorithm~\cite{sharir1981strong}. We nevertheless build on Tarjan's algorithm as it is more efficient in practice.

\medskip

\noindent 
\textbf{Tarjan's Algorithm.} Given a directed graph, Tarjan's algorithm constructs a {\em spanning forest}, that is, a set of trees that contains all the nodes of the graph, while each tree corresponds to an SCC. 
To find the trees, the algorithm runs a depth-first search starting from an arbitrary node in the graph and creates the trees by traversing the nodes in each tree and then moving to the next tree. 
During the search for nodes in a tree, the algorithm traverses two categories of edges: (i) edges that lead to new nodes, and (ii) edges that lead to already visited nodes. The edges of the first category are called {\em tree edges} as they add new nodes to the tree. The edges of the second category are classified into three subcategories: (ii.a) the edges that move from an ancestor in the tree to one of its descendants, which are ignored by the algorithm, (ii.b) the edges that move from a descendant to its ancestor, called {\em fronds}, and (ii.c) the edges that move from one subtree to another subtree in the same tree, called {\em cross-links}.
Tarjan's algorithm assigns numbers to the nodes in the order they are visited during the search (i.e., a node with a smaller number is visited first) and pushes them in a stack. This stack, which we call {\em SCC stack}, differs from the stack used to implement the depth-first search. A visited node is removed from the search stack during backtracking when all its child nodes are visited in the depth-first search. However, such a node remains in the SCC stack as long as there is a path from it to a node below it in the stack. This allows the SCC stack to keep track of the nodes in the current tree. 
The algorithm uses the assigned numbers to understand whether a node is a root node after visiting all its children. This is done by checking if the traversed edges, while visiting the descendant nodes, are fronds and cross-links. After finding a root, the algorithm creates an SCC by popping nodes from the top of the SCC stack and adding them to the SCC until it reaches the root of the current tree. The algorithm continues until all the nodes in the graph are visited and returns the obtained trees (i.e., the SCCs of the graph). 


\medskip
\noindent \textbf{The Procedure FindSpecialSCC.} It is a simple extension of Tarjan's algorithm for finding the special SCCs in a dependency graph. Such an extension is needed as we need a mechanism that allows us to check whether a SCC obtained by Tarjan's algorithm is indeed special. This is done by pushing a dummy token in the stack (the stack that stores the visited nodes in the current component) whenever the algorithm traverses a special edge. Note that the algorithm pushes this dummy token even if the special edge is ignored, as we explained in Targan's algorithm. After finding the root of the current SCC and popping the nodes to create the SCC, the algorithm labels the SCC as special if there is a dummy token between the popped nodes. Only the special SCCs are stored and returned.

\subsection{Check for Positions Support} \label{sec:support}

The procedure $\mathsf{Supports}$ consists of two steps: (1) query the database to find the positions of the extensional predicates, and (2) traverse the dependency graph starting from the positions in the special SCCs in the reverse order to reach the positions computed in the first step. 
Step (1) has been implemented via a single SQL query that returns the list of non-empty relations, which we then use to create the set of positions of the extensional predicates, denoted $P_D$. The SQL query has been implemented using the catalog of the DBMS that stores the database, which allows the query to find the set of relations names faster and without accessing the actual data.
For step (2), we start from the given set of positions $P$, i.e., the set of positions in the special SCCs, and traverse the graph in the reverse order using the reverse links in the adjacency list of the dependency graph, as explained in Section~\ref{sec:bgarph}. The procedure returns \true if the graph traversal in the second step reaches a node (i.e., a position) in $P_D$ from an extensional predicate; otherwise, it returns \false.

\subsection{Dynamic Simplification}\label{sec:d-simplification} 

The procedure $\mathsf{DynSmplification}$ takes as input a database $D$ and a set $\dep$ of linear TGDs, and returns the set $\dsimple{D}{\dep}$ of simple-linear TGDs. It is an iterative procedure that uses two sub-procedures: $\mathsf{FindShapes}$ that computes the set of shapes $S$ of the atoms of $D$, and $\mathsf{Applicable}$ that computes the simplified TGDs using the shapes of $S$. We discuss the implementation details of those two sub-procedures.

\medskip

\noindent 
\textbf{The Procedure $\mathsf{FindShapes}$.} 
We have two kinds of implementations for this procedure, that is, {\em in-memory} and {\em in-database}, which we discuss below. In Sections~\ref{sec:linear-ex} and~\ref{sec:rw-kb}, we conduct experiments comparing the performance of the two implementations of $\mathsf{FindShapes}$, and we discuss when one outperforms the other.

\medskip
\noindent \underline{\textbf{In-memory Implementation}}
\smallskip

\noindent For the in-memory implementation, we run an SQL query for each relation $R$ of $D$ to load all the tuples of $R$ into the main memory. We then construct the set of shapes of the atoms in each relation $R$ by iterating over its tuples $\bar c$, and generating the shape of each atom $R(\bar c)$. For relations that cannot be entirely loaded into the main memory, we split them into smaller relations processed separately.

\medskip
\noindent \underline{\textbf{In-database Implementation}}
\smallskip

\noindent The in-database implementation does not load the relations, but instead runs SQL queries to find the shapes of the atoms in each relation. In particular, we translate each possible shape to a Boolean query that evaluates to \true if the shape exists in the database. The query for a shape of a predicate $R$ is of the following general form:

\vspace{2mm}
\begin{Verbatim}[commandchars=\\\{\}]
  \textcolor{blue}{SELECT CASE WHEN EXISTS} 
    (\textcolor{blue}{SELECT} * \textcolor{blue}{FROM} \textcolor{violet}{R} \textcolor{blue}{WHERE} \textcolor{violet}{Equality\_Conditions} 
                     \textcolor{blue}{AND} \textcolor{violet}{Diseqality\_Conditions})
            \textcolor{blue}{THEN} 1 \textcolor{blue}{ELSE} 0 \textcolor{blue}{END} 
\end{Verbatim}
\vspace{2mm}

\noindent where the equality and disequality conditions are consistency checks according to the given shape. For example, the shape $R_{[1,1,2]}$ translates to the following SQL query $Q$:

\vspace{2mm}
\begin{Verbatim}[commandchars=\\\{\}]
  \textcolor{blue}{SELECT CASE WHEN EXISTS} 
        (\textcolor{blue}{SELECT} * \textcolor{blue}{FROM} \textcolor{violet}{R} \textcolor{blue}{WHERE} \textcolor{violet}{a1}=\textcolor{violet}{a2} \underline{\textcolor{blue}{AND} \textcolor{violet}{a2}!=\textcolor{violet}{a3}})
            \textcolor{blue}{THEN} 1 \textcolor{blue}{ELSE} 0 \textcolor{blue}{END} 
\end{Verbatim}
\vspace{2mm}

\noindent where we assume that the predicate $R$ comes with the attributes \textcolor{violet}{\texttt{a1}}, \textcolor{violet}{\texttt{a2}}, and \textcolor{violet}{\texttt{a3}}.
Of course, running an SQL query per shape results in many queries for predicates with high arity. However, depending on the database $D$, many of these queries may be unnecessary since do not find any shapes. To avoid running some of those queries, we use the Apriori algorithm's idea to find association rules~\cite{agrawal1994fast}. We start by running queries for checking the existence of more general shapes (e.g., $R_{(1,1,2)}$) before we check more specific shapes (e.g., $R_{(1,1,1)}$). For each shape, we run a pair of queries. The first query is a relaxed version of the general query explained above without the disequality conditions. For example, the first query $Q'$ for $R_{(1,1,2)}$ is the query $Q$ above without the underlined condition. We only continue to run $Q$ if $Q'$ evaluates to \true. Additionally, we do not run the query for more specific shapes, e.g., $R_{(1,1,1)}$, if $Q'$ evaluates to \false. This allows us to avoid the execution of many queries for specific shapes by running a single query for more general shapes.

\medskip

\noindent 
\textbf{The Procedure $\mathsf{Applicable}$.} Recall that $\mathsf{Applicable}$ takes as input a set $S$ of shapes and a set $\dep$ of linear TGDs, and returns the set of TGDs of $\simple{\dep}$ such that the predicate of their body belongs to $S$. As explained in Section~\ref{sec:alg}, this is done by iterating over all TGDs $\sigma \in \dep$ of the form $R(\bar x) \ra \exists \bar z\, \psi(\bar y,\bar z)$ and homomorphisms $h$ from $\{R(\bar x)\}$ to $\mi{DB}[S]$, and collecting the simplification of $\sigma$ induced by the $h$-specialization of $\bar x$. Applying the above iterative process with a large set of shapes and a large set of TGDs can be very costly. Thus, to improve the performance of this process, the implementation uses an index structure that enables fast access to the TGDs. The index structure maps each predicate $R \in \sch{\dep}$ to the set of TGDs of $\dep$ that their body-atom uses $R$. This allows the procedure to iterate over the current shapes in each iteration and quickly access the relevant TGDs in the index.
However, checking whether a relevant TGD is applicable remains costly. This task requires checking the shapes of the body-atoms in all the simplified TGDs obtained from the TGD with the current shape. To facilitate this, we generate and store the shape of the body-atom of each TGD. Additionally, we store an array of strings representing all possible identifiers of tuples up to the maximum arity of the schema that allows us to quickly find the shapes of the body-atoms in simplified TGDs.

\section{Experimental Infrastructure}\label{sec:setting}

Our goal is to experimentally evaluate the behaviour of the termination algorithms presented in Section~\ref{sec:implementation} with the aim of clarifying whether they can be applied in a practical context, and if not, reveal their limitations. To this end, we are going to conduct extensive experiments with synthetic data and sets of TGDs. Therefore, we need a way to generate databases and sets of TGDs that are suitable for such an experimental evaluation. We proceed to discuss our tools for generating data (Section~\ref{sec:synth-data}) and sets of TGDs (Section~\ref{sec:synth-rule}). 
Let us clarify that in the rest of the paper, whenever we say that we randomly select an element from a certain space of elements, we actually mean that we select such an element uniformly at random.
%


\subsection{Data Generator}\label{sec:synth-data}


There are freely available data generators such as TPC-H\footnote{\href{https://www.tpc.org/tpch/}{https://www.tpc.org/tpch/}} and DataFiller\footnote{\href{https://github.com/memsql/datafiller}{https://github.com/memsql/datafiller}}. However, none of the existing tools is suitable for our purposes. To effectively evaluate the dynamic simplification procedure presented above, which is a key component of the termination algorithm for linear TGDs, we need to make sure that the generated database contains a variety of shapes. This is precisely the limitation of the existing data generators as they do not allow us to control the shape of the generated atoms. Hence, we had to implement our own data generator that overcomes the above limitation.

Our data generator has tuning parameters that allow us to determine key properties of the generated database $D$: the number of predicates in $D$, the minimum and maximum arity of those predicate, the size of the database domain (i.e., the number of values in $\adom{D}$), and the number of tuples in each relation of $D$.
In particular, the generator takes as input a tuple of integer values for the tuning parameters $(\mi{preds},\mi{min},\mi{max},\mi{dsize},\mi{rsize})$, and constructs a database $D$ such that, with $\ins{S} = \{R \mid R(\bar c) \in D\}$, $|\ins{S}| = \mi{preds}$, the predicates of $\ins{S}$ have arity between $\mi{min}$ and $\mi{max}$, $|\adom{D}| = \mi{dsize}$, and, for each $R \in \ins{S}$, $|\{\bar c \mid R(\bar c) \in D\}| = \mi{rsize}$.
To this end, it first generates a set $\ins{S}$ consisting of $\mi{preds}$ different predicates, and it randomly selects an arity for each such predicate from the range $[\mi{min},\mi{max}]$. It then generates a database by adding $\mi{rsize}$ tuples to each predicate of $\ins{S}$ that are formed using $\mi{dsize}$ different constant values. Now, to ensure that the obtained database contains different shapes, each tuple is generated by randomly selecting a shape and filling the positions by randomly picking values from the database domain of size $\mi{dsize}$ without repetition, that is, a shape determines how many times the same value is repeated in a tuple.

%

\subsection{TGD Generator}\label{sec:synth-rule}

As for the data generator, existing TGD generators (see, for example,~\cite{arocena2015ibench,benedikt2017benchmarking}) are not suitable for our purposes since they do not allow us to control the shape of the atoms occurring in the bodies of the generated TGDs, which is crucial for generating sets of TGDs that are suitable for our experiments. Thus, we had to implement our own TGD generator that supports this key feature.

Our TGD generator has tuning parameters that allows us to determine key properties of the generated set $\dep$ of TGDs: the size of the schema (that is, the size of $\sch{\dep}$), the minimum and maximum arity of the predicates of $\sch{\dep}$, the number of TGDs in $\dep$, and the underlying class of $\dep$ (that is, whether $\dep$ is a set of simple-linear or just linear TGDs). 
In particular, the TGD generator takes as input a set $\ins{S}$ of predicates and a tuple of values for the tuning parameters $(\mi{ssize},\mi{min},\mi{max},\mi{tsize},\mi{tclass})$, and constructs a set $\dep$ of TGDs such that $\sch{\dep} \subseteq \ins{S}$, $|\sch{\dep}| = \mi{ssize}$, the predicates of $\sch{\dep}$ have arity between $\mi{min}$ and $\mi{max}$, $|\dep| = \mi{tsize}$, and $\dep$ falls in the class $\mi{tclass}$.
To this end, it first chooses a subset $\ins{S}'$ of $\ins{S}$ such that $|\ins{S}'| = \mi{ssize}$ and its predicates have arity between $\mi{min}$ and $\mi{max}$, and then generates the desired set of simple-linear or linear TGDs using all the predicates of $\ins{S}'$; the actual generation is described below.

Let us stress that in our experiments we consider only single-head TGDs, that is, TGDs with only one head-atom, despite the fact that our termination algorithms work with multi-head TGDs, that is, TGDs that have several atoms in their heads.
This is because the number of atoms occurring in the head of a TGD is typically negligible compared to the number of TGDs, and having several atoms in the heads of TGDs does not affect the number of shapes occurring in the database.
As we shall see in Sections~\ref{sec:slinear-ex} and~\ref{sec:linear-ex}, the number of TGDs and the number of database shapes are the main parameters impacting the runtime of the algorithms, and thus, our experimental evaluation provides conclusive results even if we consider single-head TGDs. Hence, for clarity, our TGD generator described below is designed to generate single-head TGDs.



\medskip
\noindent \underline{\textbf{Simple-Linear TGDs}}
\smallskip

\noindent To generate a simple-linear TGD, the generator randomly selects two predicates from $\ins{S}'$ that will be used for forming the body- and the head-atom, respectively. The random selection is with repetition to allow the same predicate to appear in both the body and the head of the TGD. Since we target a simple-linear TGD, we use different variables to fill the positions in the body-atom. Now, for the head-atom, we fill each position with either an existentially quantified variable, or a universally quantified variable that has been already used in the body-atom. In particular, for each position $\pi$ of the head-atom, the generator classifies $\pi$ as an existential position with probability $10\%$. In this case, $\pi$ is filled with a fresh variable that does not appear in the body-atom; otherwise, it is filled with a randomly selected variable from the body-atom.

\medskip
\noindent \underline{\textbf{Linear TGDs}}
\smallskip

\noindent Generating linear TGDs is done in the same way as for simple-linear TGDs with the crucial difference that, after randomly selecting the predicates of the body- and the head-atom, the generator randomly chooses a shape for the body-atom and then applies similar steps as for simple-linear TGDs to fill the positions of the body- and the head-atom. Note that the selection of the body-variables is guided by the chosen shape, which in turn allows for the repetition of variables in the body-atom of the generated linear TGD.

\medskip

We are now ready to proceed with our experimental evaluation. Note that for the experiments we used a server with an Intel Core i5 3.00GHz CPU and 16GB RAM, all the databases in our experiments are stored in a PostgreSQL 11.5 instance, and the termination algorithms in question have been implemented in Java SE 11.






\section{Evaluation for Simple-Linear TGDS}\label{sec:slinear-ex}

We start with the experimental evaluation of $\mathsf{IsChaseFinite[SL]}$, which is depicted in Algorithm~\ref{alg:slinear}. Towards a refined analysis, we are going to break down its end-to-end runtime, which we denote by \ttotal, into the following three time parameters:

\begin{itemize}
    \item \tparse: time to parse the TGDs from an input file,
    \item \tgraph: time to build the dependency graph $G$ of the input set of TGDs, and
    \item \tcomp: time to find the special SCCs in the graph $G$.
\end{itemize}

\noindent In the rest of the section, we explain how we generate the sets of simple-linear TGDs that are used in our experimental evaluation, and then present our experimental results and discuss the take-home messages. But let us first give a couple of clarification remarks.

\medskip

\noindent 
\textbf{Remark 1.} In our analysis, we neglect the time taken by the procedure $\mathsf{Supports}$, which, as explained in Section~\ref{sec:support}, consists of two steps: (1) find the predicates occurring in the input database, and (2) traverse the dependency graph starting from the positions in the special SCCs in the reverse order to reach the positions of the relations computed in the first step.
Step (1) is performed by running a fast query on the catalog of the DBMS storing the input database, and can be safely ignored as it does not impact the rest of the algorithm. 
Concerning step (2), the time to traverse the dependency graph is negligible compared to the time needed to find the special SCCs, which is already in the order of milliseconds.
Summing up, the procedure $\mathsf{Supports}$ takes insignificant time compared to the rest of the algorithm, which we can simply ignore without affecting our analysis for $\mathsf{IsChaseFinite[SL]}$. Therefore, in our experiments, we assume that all the predicates used by the set of simple-linear TGDs occur in the database, and thus, all the positions in the special SCCs are trivially supported. 
This in turn simplifies our experiments as they can be conducted using a very simple database that can be induced by the set $\dep$ of simple-linear TGDs, denoted $D_\dep$, without using our data generator discussed above. In fact, $D_\dep$ has an atom $R(c_1,\ldots,c_n)$, where $c_1,\ldots,c_n$ are distinct constants, for each predicate $R \in \sch{\dep}$.

\medskip

\noindent 
\textbf{Remark 2.} The parsing time (\tparse) clearly depends on the parser's performance, and is not so crucial for evaluating the performance of $\mathsf{IsChaseFinite[SL]}$. However, we include it in our analysis in order to have an accurate figure for the end-to-end runtime of the algorithm, and also compare it with the other two time parameters.

\ignore{\begin{table}[h]
\centering
\begin{tabular}{lc}
\toprule
Parameter & Description \\
\midrule
\ttotal & Total end-to-end runtime
\\
\tparse & Time for parsing a set of tgds \\
\tgraph & Time for building a dependency graph \\
\tcomp & Time for finding special SCCs \\
\nrules &  The number of tgds in a set of tgds\\
\nshape & Time for parsing a set of tgds 
\bottomrule
\end{tabular}
\vspace{0.3cm}
\caption{Evaluation parameters}
\vspace{-3mm}
\label{tab:params}
\end{table}}

\subsection{Generating Simple-Linear TGDs}

We now discuss how the sets of simple-linear TGDs used in our experiments are generated.
To systematically generate a representative family of sets of TGDs, without favouring any of the two key parameters, namely the size of the underlying schema and the number of TGDs, we consider three {\em predicate profiles} consisting of sets of TGDs that mention [5,200], [200,400], and [400,600] predicates of arity between 1 and 5, and we further consider three {\em TGD profiles} consisting of sets of TGDs with [1,333K], [333K,666K], and [666K,1M] TGDs. Note that our choice to fix the arity of the predicates between 1 and 5 is consistent with what we observe in real-life scenarios, where the arity is typically small.
The combination of those predicate and TGD profiles gives rise to nine {\em combined profiles} consisting of sets of TGDs with similar syntactic properties. 
For example, the combined profile obtained from the predicate profile [200,400] and the TGD profile [333K,666K] consists of sets of TGDs $\dep$ such that $200 \leq \sch{\dep} \leq 400$, each predicate of $\sch{\dep}$ has arity between 1 and 5, and $333{\rm K} \leq |\dep| \leq 666{\rm K}$.
For our experiments, we generated 100 sets of TGDs for each of the nine combined profiles, totalling 900 sets of simple-linear TGDs. This was done as follows.
We have first constructed the underlying schema $\ins{S}$ by generating 1000 predicates, while their arities were randomly selected from the range [1,5]. Then, for the combined profile induced by the predicate profile $[x,y]$ and the TGD profile $[z,w]$, we have generated 100 sets of simple-linear TGDs by repeatedly executing our TGD generator with input the schema $\ins{S}$ and the tuple of values for the tuning parameters $(\mi{ssize},1,5,\mi{tsize},\mathsf{SL})$, where $\mi{ssize}$ and $\mi{tsize}$ were randomly chosen from the range of values $[x,y]$ and $[z,w]$, respectively.




\begin{figure}[t]
	\centering
	\begin{subfigure}[b]{0.23\textwidth}
		\centering
		\includegraphics[width=\textwidth]{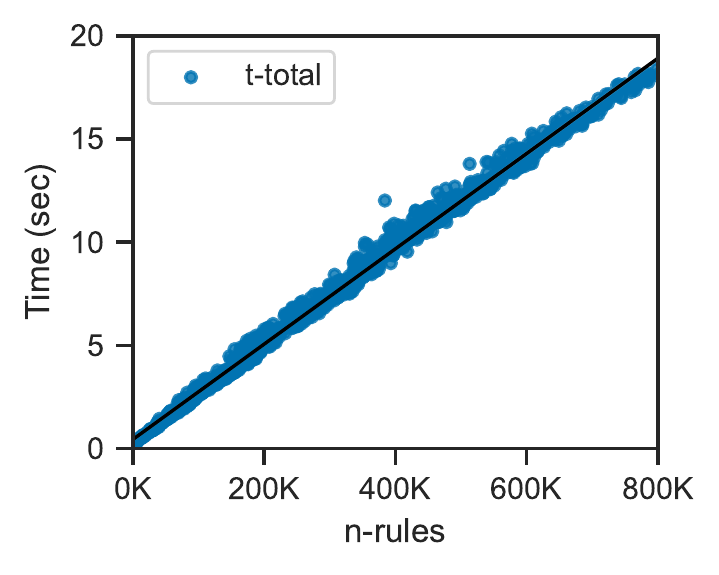}
		\caption{\ttotal}
		\label{fig:slinear-ttotal}
	\end{subfigure}
	\hfill
	\medskip
	\begin{subfigure}[b]{0.23\textwidth}
		\centering
		\includegraphics[width=\textwidth]{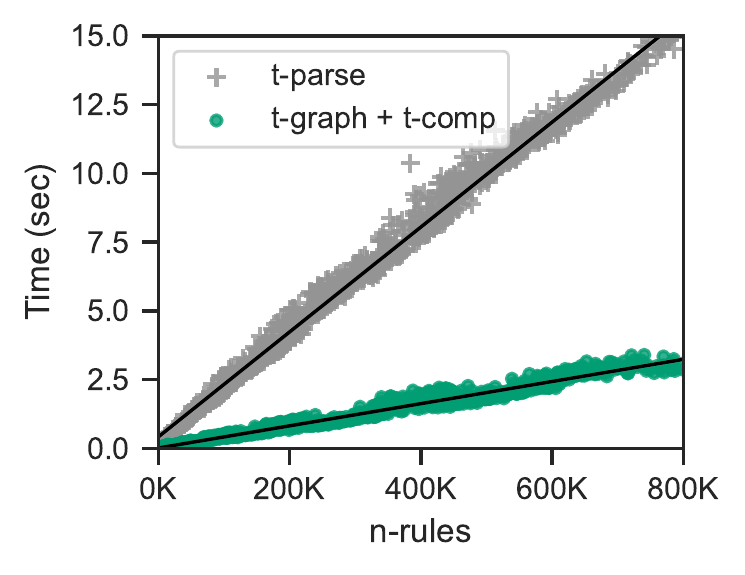}
		\caption{\tparse vs. \tgraph+\tcomp}
		\label{fig:slinear-parse-graphcomponent}
	\end{subfigure}
	\hfill
	\begin{subfigure}[b]{0.23\textwidth}
		\centering
		\includegraphics[width=\textwidth]{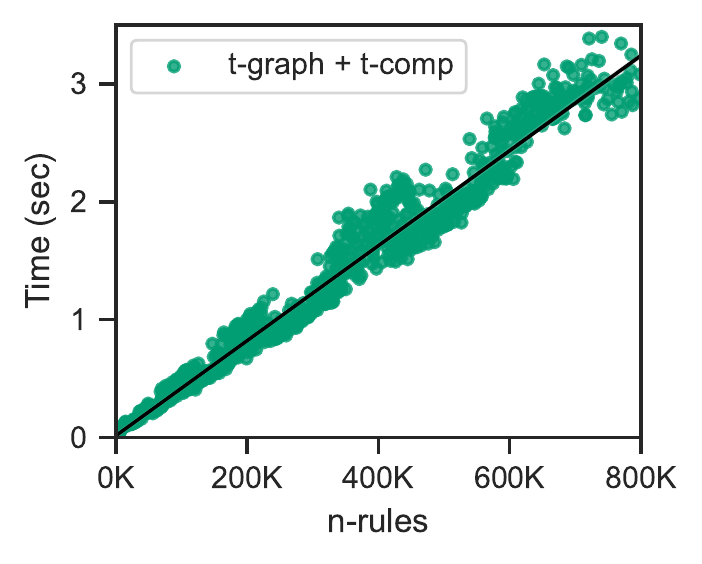}
		\caption{\tgraph + \tcomp}
		\label{fig:slinear-graphcomponent}
	\end{subfigure}
	\hfill
	\begin{subfigure}[b]{0.23\textwidth}
		\centering
		\includegraphics[width=\textwidth]{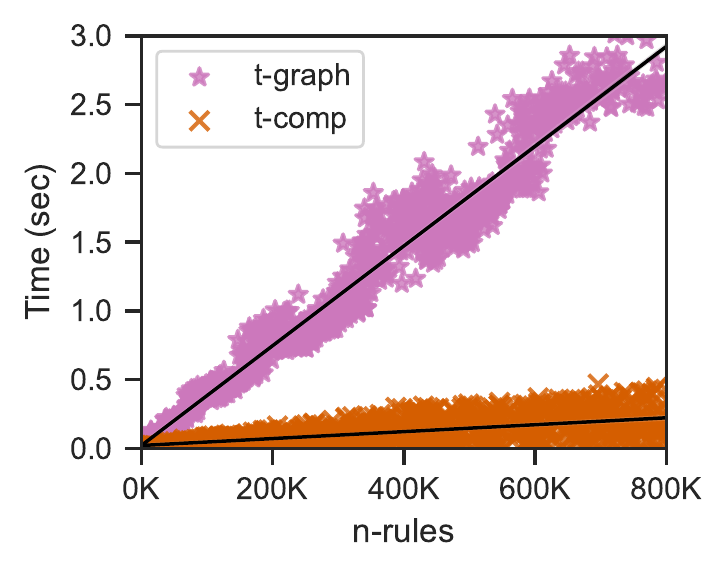}
		\caption{\tgraph vs. \tcomp}
		\label{fig:slinear-graph-component}
	\end{subfigure}
	\caption{Runtime of $\mathsf{IsChaseFinite[SL]}$.}
	\label{fig:slinear-time}
\end{figure}

\subsection{Experimental Evaluation}



The algorithm $\mathsf{IsChaseFinite[SL]}$ was run for each one of the 900 sets of TGDs of the combined profiles discussed above. Recall that the input database is induced by the input set of TGDs, i.e., for the set of TGDs $\dep$, the database is $D_\dep$.
The scatter plots in Figure~\ref{fig:slinear-time} show the runtime of $\mathsf{IsChaseFinite[SL]}(D_\dep,\dep)$, for each set $\dep$ of TGDs from the combined profiles. In particular, each point in the plots corresponds to one of the 900 sets of TGDs. 
Figure~\ref{fig:slinear-ttotal} shows the total runtime (\ttotal) for sets of TGDs with various sizes (\nrule). Figure~\ref{fig:slinear-parse-graphcomponent} breaks down \ttotal into the time to parse the TGDs (\tparse) and the time to build their dependency graph and find the special SCCs (\tgraph + \tcomp). Figure~\ref{fig:slinear-graphcomponent} zooms in \tgraph + \tcomp, which are shown separately in Figure~\ref{fig:slinear-graph-component}.

It is evident from the above scatter plots that the time parameters \tparse and \tgraph increase linearly as long as we increase \nrule, whereas \tcomp increases very slowly.
Let us remark that we have not observed such a linear relationship (in fact, we have not observed any correlation) between the time parameters \tparse and \tgraph, and the number of predicates of the underlying schema.
The linear relationship between \tparse and \nrule is because parsing each TGD takes constant time since the arity of the predicates falls in the limited range [1,5], and each TGD has one atom in its body and one atom in its head. Note that allowing multi-heads will not change this since, as discussed in Section~\ref{sec:setting}, the number of head-atoms is negligible compared to the number of TGDs.
The linear relationship with \tgraph (as shown in Figure~\ref{fig:slinear-graph-component}) is because the algorithm iterates over the TGDs and spends constant time to process each TGD and update the graph by adding new nodes and edges. Again, since the arity of the predicates falls in [1,5], and each TGD has one atom in its body and one atom in its head, the number of nodes and edges added in the dependency graph due to a certain TGD is in a small fixed range, and thus, the time to update the graph w.r.t. each TGD is constant.
The fact that \tcomp increases very slowly is because finding the special SCCs solely depends on the dependency graph, which is in general much smaller than the set of TGDs, while Tarjan's algorithm is quite efficient that runs in linear time in the size of the underlying graph.


\subsection{Take-home Messages} \label{sec:discussion-simple}

The main takeaway from the experimental results for simple-linear TGDs is that the primary parameter impacting the runtime of $\mathsf{IsChaseFinite[SL]}$ is the number of TGDs (\nrule), and we have also observed that the algorithm is very fast even for extremely large sets of TGDs. In fact, most of the end-to-end runtime is spent on parsing (\tparse) and building the dependency graph (\tgraph), whereas the time to the find special SCCs (\tcomp) is insignificant compared to \tparse and \tgraph. To be more precise, \tparse is much larger than \tgraph, and it actually takes most of the total end-to-end runtime of the algorithm. This illustrates the effectiveness of $\mathsf{IsChaseFinite[SL]}$ as the actual check for the finiteness of the chase instance for large sets of TGDs is much faster than even reading and parsing the TGDs from the input file.


\section{Evaluation for Linear TGDs} \label{sec:linear-ex}

We now proceed with the evaluation of $\mathsf{IsChaseFinite[L]}$, depicted in Algorithm~\ref{alg:dterm}. Differently from $\mathsf{IsChaseFinite[SL]}$, where the input database did not play any crucial role, we now have a component that heavily relies on the database, that is, the procedure that computes the database shapes, which is part of dynamic simplification.
In other words, we have the {\em database-dependent component} of $\mathsf{IsChaseFinite[L]}$, that is, find the database shapes, and the {\em database-independent component}, that is, simplify the given set of linear TGDs by using the database shapes, build the dependency graph of the simplified set of TGDs, and find the special SCCs in this graph.
We claim that these two components, which from now on we call db-dependent and db-independent, respectively, should be evaluated separately as their runtime is impacted by different parameters.
Concerning the db-dependent component, it is obvious that it is only affected by the database, whereas the set of TGDs plays no role. 
On the other hand, although it is clear that the db-independent component is affected by the set of TGDs, it is not straightforward to see that it is not affected by the input database since it operates on a dynamically simplified set of TGDs. Interestingly, we experimentally confirm below that this is indeed the case.
Consequently, towards a refined analysis of the algorithm $\mathsf{IsChaseFinite[L]}$, we are going to consider the following four time parameters:
\begin{itemize}
	\item \tshape: time to find the database shapes,
	\item \tparse: time to parse the TGDs from an input file,
	\item \tgraph: time to build the dependency graph $G$ of the simplified version of the input set of TGDs (including the time for the simplification using the database shapes), and
	\item \tcomp: time to find the special SCCs in the graph $G$.
\end{itemize}
Clearly, \tshape refers to the runtime of the db-dependent component, whereas \tparse + \tgraph + \tcomp, which we denote by \ttotal, refers to the end-to-end runtime of the db-independent component.
In the rest of the section, we explain how we generate the databases and the sets of linear TGDs that are used in our experimental evaluation, confirm that the db-independent component is not affected by the input database, and then present our experimental results for the two components of $\mathsf{IsChaseFinite[L]}$ and discuss the take-home messages.
Note that, as stated in Remark 2 in Section~\ref{sec:slinear-ex}, although the parsing time is not crucial for evaluating the performance of the db-independent component, we consider it in order to have an accurate figure for the end-to-end runtime, and also compare it with \tgraph and \tcomp.

\subsection{Generating Databases and Linear TGDs}


The goal is to generate a family of pairs of the form $(D,\dep)$, where $D$ is a database and $\dep$ a set of linear TGDs, that will serve as the input to $\mathsf{IsChaseFinite[L]}$ during the experimental evaluation.
To this end, we first constructed a very large database, which we dub $D^\star$, by using our data generator. In particular, we called the data generator with input $(1000,1,5,500{\rm K},500{\rm K})$, and obtained $D^\star$ that mentions 1000 predicates of arity between 1 and 5, and each such predicate has 500K tuples, resulting in a very large database with 500M tuples in total.
We further devised views over $D^\star$ that allow us to define on-demand virtual databases with 1K, 50K, 100K, 250K, and 500K tuples per predicate, resulting in databases with 1M, 50M, 100M, 250M, and 500M tuples in total, respectively. Those views are actually implemented as a simple SQL query that simply keeps the first 1K, 50K, 100K, 250K, and 500K tuples per predicate, respectively. Note that the tuples in $D^\star$ are lexicographically sorted, which means that the different shapes in a relation of $D^\star$ are evenly distributed. This is turn ensures that the virtual databases defined via the views have a variety of shapes, which is crucial for our purposes.
Having the database $D^\star$ and the database views in place the desired family of pairs was generated as described below.

For each one of the nine combined profiles used in the generation of simple-linear TGDs in Section~\ref{sec:slinear-ex}, we generated 5 sets of linear TGDs, totalling 45 sets. In particular, for the combined profile induced by the predicate profile $[x,y]$ and the TGD profile $[z,w]$, we have generated 5 sets of linear TGDs by repeatedly executing our TGD generator with input the schema $\{R \mid R(\bar c) \in D^\star\}$, i.e., the 1000 predicates occurring in $D^\star$, and the tuple of values for the tuning parameters $(\mi{ssize},1,5,\mi{tsize},\mathsf{L})$, where $\mi{ssize}$ and $\mi{tsize}$ were randomly chosen from the range of values $[x,y]$ and $[z,w]$, respectively. Let $\dep^\star$ be the family that collects the 45 generated sets of linear TGDs.
Then, for each set $\dep \in \dep^\star$ of linear TGDs, by exploiting the database views discussed above, we obtained five virtual databases of varying size (1K, 50K, 100K, 250K, and 500K tuples per predicate), denoted $D_{\dep}^{1}$, $D_{\dep}^{50}$, $D_{\dep}^{100}$, $D_{\dep}^{250}$, and $D_{\dep}^{500}$, respectively, leading to five pairs.
Summing up, we generated the family
\[
\left\{\left(D_{\dep}^{s},\dep\right) \mid \dep \in \dep^\star \text{ and } s \in \{1,50,100,250,500\}\right\}
\]
consisting of 225 pairs that will serve as the input to $\mathsf{IsChaseFinite[L]}$.

\begin{figure*}[h!]
	\begin{minipage}{.74\textwidth}
	\centering
	\begin{subfigure}[b]{0.3\textwidth}
		\centering
		\includegraphics[width=\textwidth]{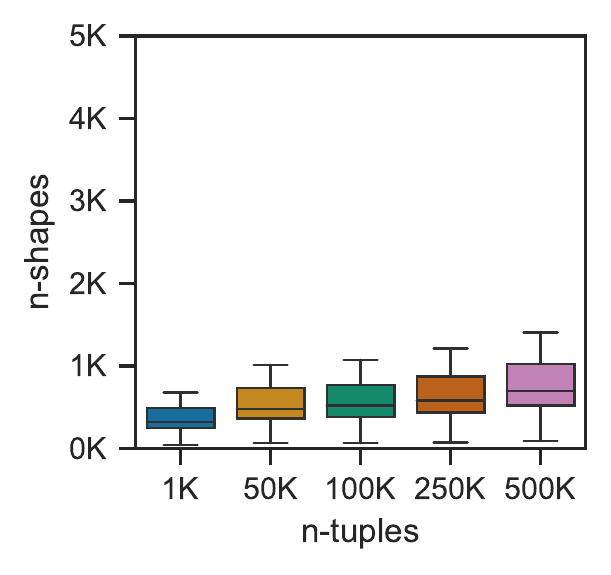}
		\vspace{-5mm}
		\caption{[5,200]}
		\label{fig:shape-5}
	\end{subfigure}
	\hspace{0.5cm}
	\begin{subfigure}[b]{0.3\textwidth}
		\centering
		\includegraphics[width=\textwidth]{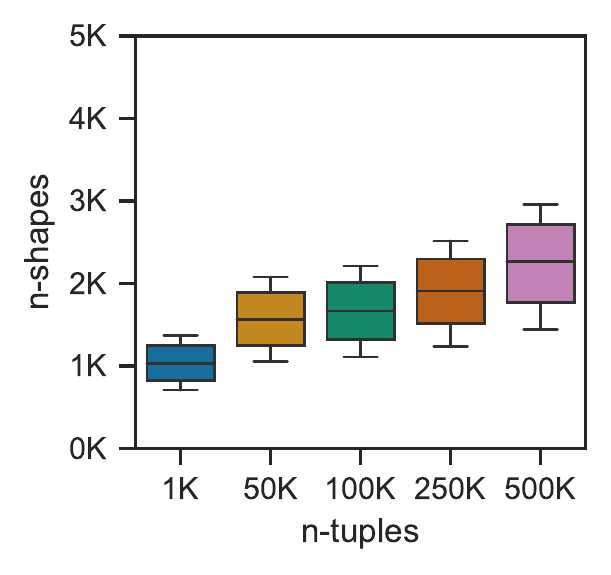}
		\vspace{-5mm}
		\caption{[200,400]}
		\label{fig:shape-200}
	\end{subfigure}
	\hspace{0.1cm}
	\begin{subfigure}[b]{0.3\textwidth}
		\centering
		\includegraphics[width=\textwidth]{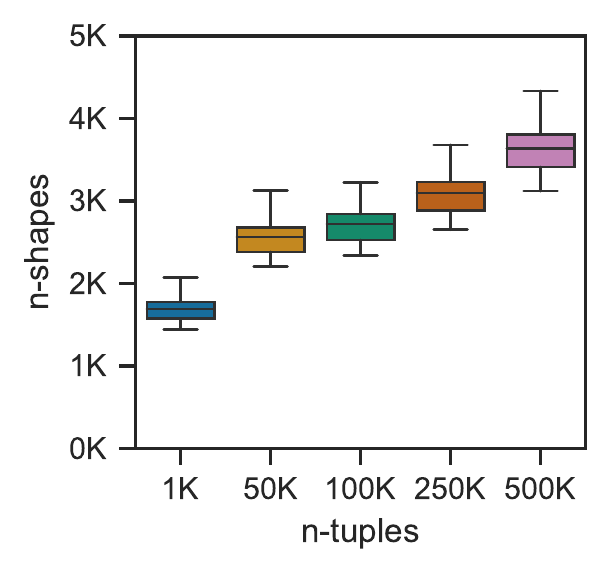}
		\vspace{-5mm}
		\caption{[400,600]}
		\label{fig:shape-400}
	\end{subfigure}
	\vspace{-3mm}
	\caption{Number of Shapes.}
	\label{fig:nshapes}
	\end{minipage}
	\begin{minipage}{.74\textwidth}
		\centering
		\begin{subfigure}[b]{0.3\textwidth}
			\centering
			\includegraphics[width=\textwidth]{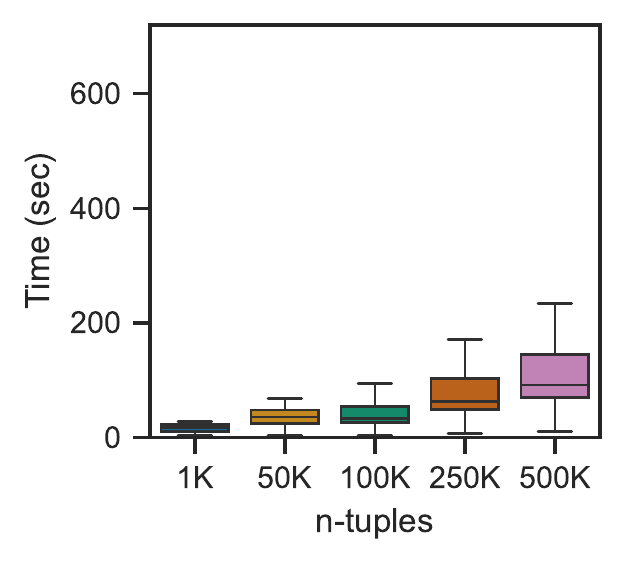}
			\vspace{-5mm}
			\caption{[5,200]}
			\label{fig:shape-time-5}
		\end{subfigure}
		\hspace{0.5cm}
		\begin{subfigure}[b]{0.3\textwidth}
			\centering
			\includegraphics[width=\textwidth]{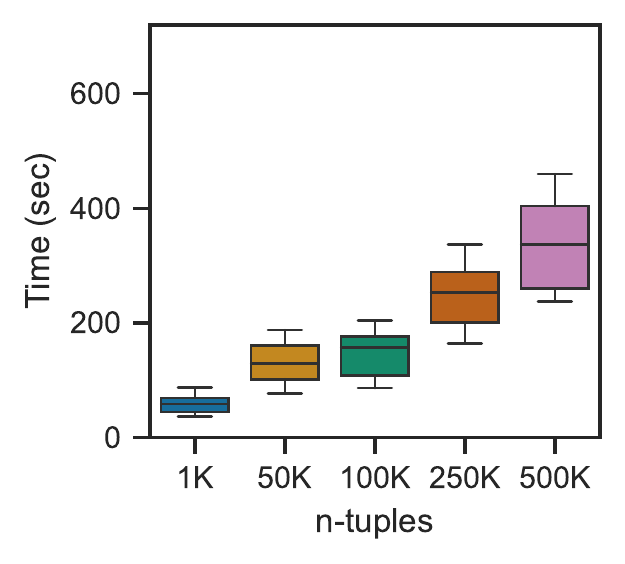}
			\vspace{-5mm}
			\caption{[200,400]}
			\label{fig:shape-time-200}
		\end{subfigure}
		\hspace{0.1cm}
		\begin{subfigure}[b]{0.3\textwidth}
			\centering
			\includegraphics[width=\textwidth]{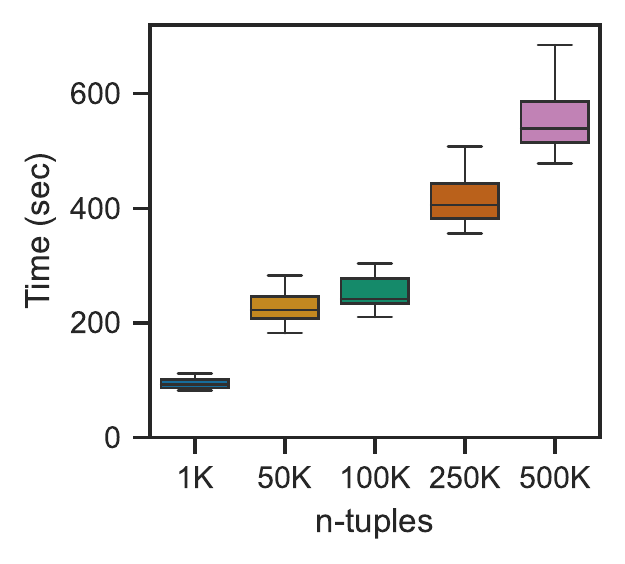}
			\vspace{-5mm}
			\caption{[400,600]}
			\label{fig:shape-time-400}
		\end{subfigure}
		\vspace{-3mm}
		\caption{Runtime of $\mathsf{FindShapes}$ (in-memory implementation).}
		\label{fig:tshapes}
	\end{minipage}
	\begin{minipage}{.74\textwidth}
		\centering
		\begin{subfigure}[b]{0.3\textwidth}
			\centering
			\includegraphics[width=\textwidth]{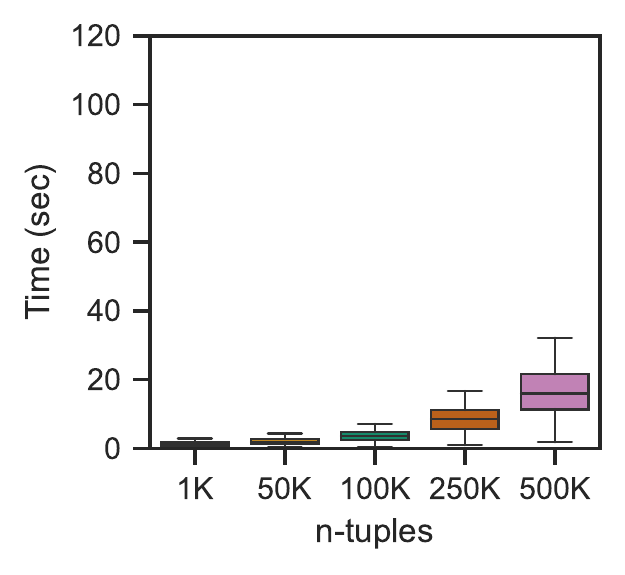}
			\vspace{-5mm}
			\caption{[5,200]}
			\label{fig:shape-time-5-q}
		\end{subfigure}
		\hspace{0.5cm}
		\begin{subfigure}[b]{0.3\textwidth}
			\centering
			\includegraphics[width=\textwidth]{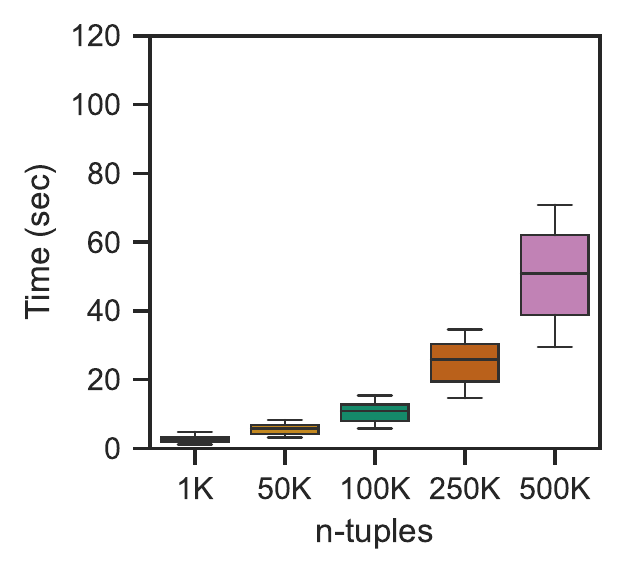}
			\vspace{-5mm}
			\caption{[200,400]}
			\label{fig:shape-time-200-q}
		\end{subfigure}
		\hspace{0.1cm}
		\begin{subfigure}[b]{0.3\textwidth}
			\centering
			\includegraphics[width=\textwidth]{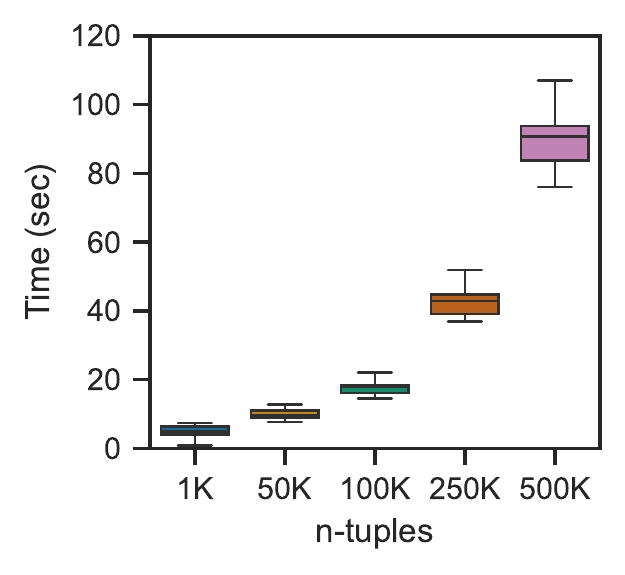}
			\vspace{-5mm}
			\caption{[400,600]}
			\label{fig:shape-time-400-q}
		\end{subfigure}
		\vspace{-3mm}
		\caption{Runtime of $\mathsf{FindShapes}$ (in-database implementation).}
		\label{fig:tshapes-q}
	\end{minipage}
\end{figure*}

\subsection{Experimental Evaluation}

Before delving into the evaluation of the two components of the algorithm $\mathsf{IsChaseFinite[L]}$, let us first confirm that indeed the db-independent component is not affected by the input database, which in turn justifies our decision to separately evaluate the two components as their runtime is impacted by different parameters.

\medskip

\noindent 
\textbf{Separate the Two Components.} The figure below depicts the average time over all generated pairs, consisting of a database $D_{\dep}^{s}$ of a certain size $s \in \{1,50,100,250,500\}$ and a set $\dep$ of linear TGDs, for building the dependency graph of the dynamically simplified version of $\dep$ using the shapes of $D_{\dep}^{s}$ and finding the special SCCs:


\centerline{\includegraphics[width=.23\textwidth]{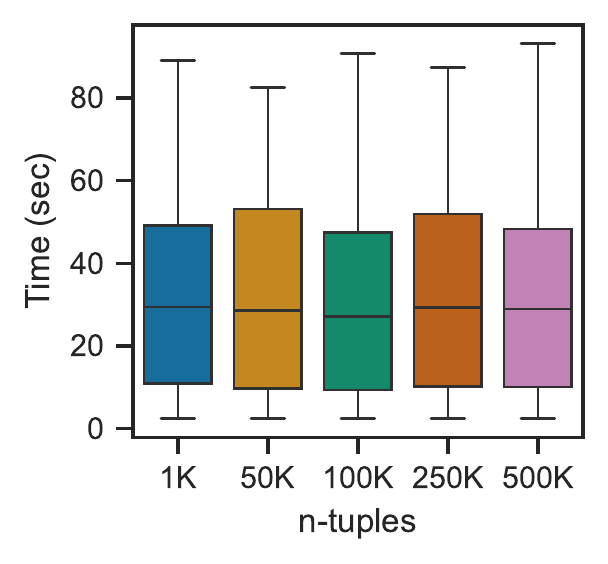}}


\noindent Interestingly, it confirms that the database size does not impact the time to build the dependency graph and find the special SCCs; thus, it does not impact the end-to-end runtime of the db-independent component, as claimed above.
This can be explained by the fact that the number of shapes in a database increases very slowly as we increase the size of the database; this is illustrated in Figure~\ref{fig:nshapes}. In particular, the bar plots in Figure~\ref{fig:nshapes} show the average number of shapes over all databases $D_{\dep}^{s}$ of a certain size $s$, where each plot corresponds to a certain predicate profile $[x,y]$, i.e., $\dep$ falls in the predicate profile $[x,y]$. It is clear that the number of shapes increases as we increase the size of the database, which is somehow expected as an increase in the database size leads to more tuples that are likely to induce more unique shapes. The interesting outcome, however, is the fact that this increase is very slow, which should be attributed to the fact that many tuples are likely to induce the same shape. 
Moreover, a new shape gives rise to only a few simplified TGDs, and it does not significantly affect the time for building and processing the dependency graph.

Let us finally observe that, by comparing the three bar plots, it is clear that the number of predicates, reflected in the predicate profile, impacts the number of shapes, which is rather expected as with more predicates there would be more shapes. This means that the number of predicates of the underlying schema is a parameter that affects the number of shapes, which explains why in our analysis above we had to separately consider the three predicate profiles.

\medskip

\noindent 
\textbf{Evaluation of the DB-dependent Component.}
We run the procedure $\mathsf{FindShapes}$ for each one of the databases $D_{\dep}^{s}$, where $\dep \in \dep^\star$ and $s \in \{1,50,100,250,500\}$; 225 executions in total. Recall that we have two kinds of implementations for the procedure $\mathsf{FindShapes}$, namely in-memory and in-database (see Section~\ref{sec:d-simplification}).
The bar plots in Figures~\ref{fig:tshapes} and~\ref{fig:tshapes-q} show the average runtime over all databases $D_{\dep}^{s}$ of a certain size $s$ for finding the shapes in the case of the in-memory and in-database implementation, respectively, where each plot corresponds to a certain predicate profile. 
Note that for both implementations we observed a similar trend, with the in-database implementation outperforming the in-memory one. 

It is evident from the bar plots in Figure~\ref{fig:tshapes-q} that the time to find the shapes increases while the database size increases, which is not surprising since, as discussed above, the number of shapes increases while the database size increases. Observe, however, that the time to find the shapes grows much faster than the actual number shapes, which should be attributed to the fact that for finding the shapes we actually need to scan the whole database.
Let us finally observe that, by comparing the three bar plots, it is apparent that the number of predicates, reflected in the underlying predicate profile, also impacts the time to find the shapes, which explains why we had to analyze each predicate profile separately.

\begin{figure}[t]
	\centering
	\begin{subfigure}[b]{0.23\textwidth}
		\centering
		\includegraphics[width=\textwidth]{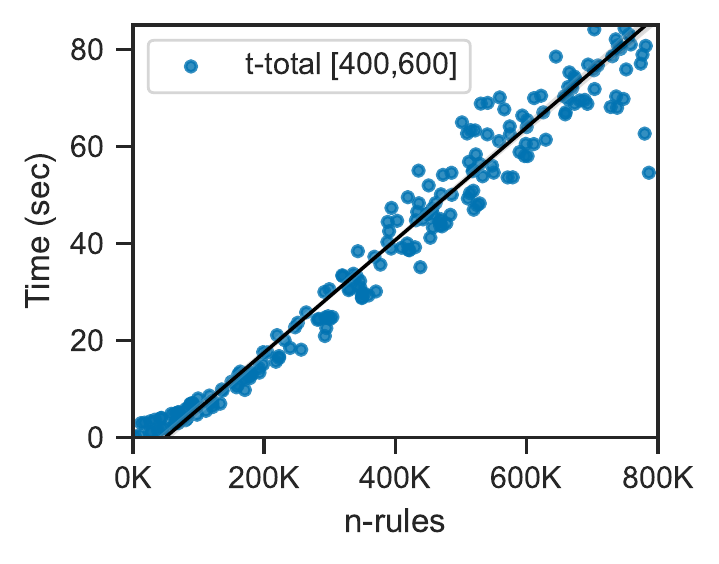}
		\caption{\ttotal}
		\label{fig:linear-total}
	\end{subfigure}
	\hfill
	\begin{subfigure}[b]{0.23\textwidth}
		\centering
		\includegraphics[width=\textwidth]{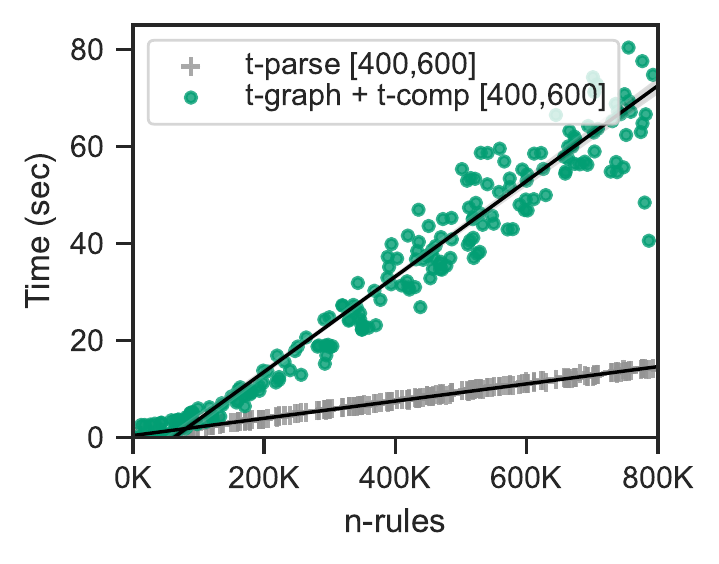}
		\caption{\tparse vs. \tgraph+\tcomp}
		\label{fig:linear-parse-graphcomponent-d3}
	\end{subfigure}
	\hfill
	\begin{subfigure}[b]{0.23\textwidth}
		\centering
		\includegraphics[width=\textwidth]{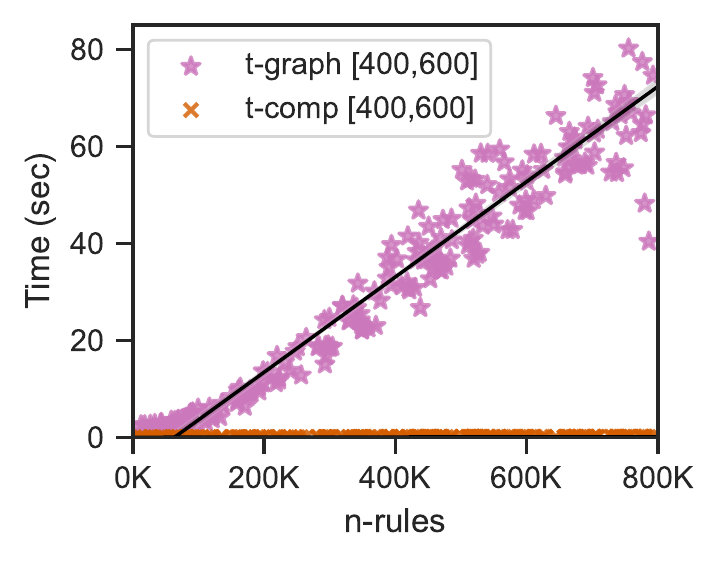}
		\caption{\tgraph vs. \tcomp}
		\label{fig:linear-graph-component-d3}
	\end{subfigure}
	\hfill
	\begin{subfigure}[b]{0.23\textwidth}
		\centering
		\includegraphics[width=\textwidth]{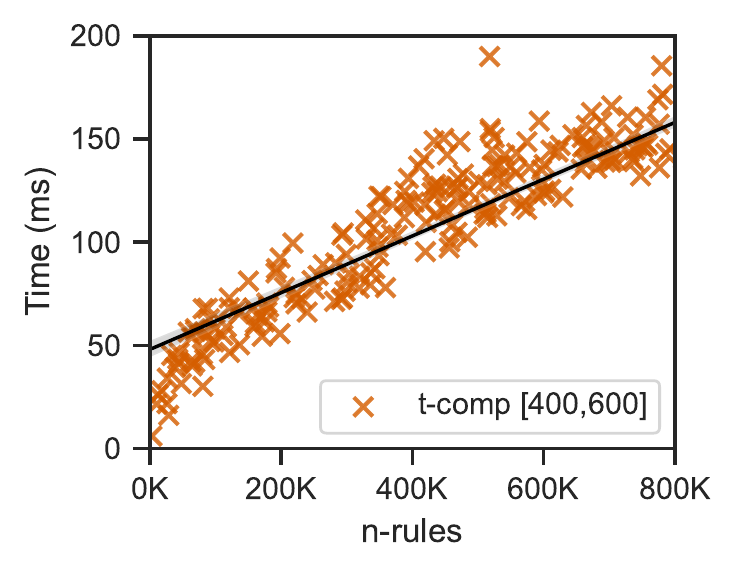}
		\caption{\tcomp}
		\label{fig:linear-graphcomponent-d3}
	\end{subfigure}
	\caption{Runtime of the db-independent component.}
	\label{fig:linear-time}
\end{figure}

\medskip

\noindent 
\textbf{Evaluation of the DB-independent Component.}
The scatter plots in Figure~\ref{fig:linear-time} show the runtime of the db-independent component of the algorithm $\mathsf{IsChaseFinite[L]}$ when executed with input $(D_{\dep}^{s},\dep)$, for each set $\dep \in \dep^\star$ falling in the predicate profile [400,600] and $s \in \{1,50,100,250,500\}$. In particular, each point in the plots corresponds to a pair $(D_{\dep}^{s},\dep)$. 
Let us stress that, unlike the analogous Figure~\ref{fig:slinear-time} for simple-linear TGDs, we focus on a particular predicate profile since otherwise we do not obtain any trend of the runtime w.r.t. the number of TGDs. In other words, the apparent linear trend observed in those plots only holds for sets of TGDs from the same predicate profile. 
This is because the number of predicates of the underlying schema impacts the number of shapes, which in turn affects the process of dynamic simplification and the size of the dependency graph, and thus, the time parameters \tgraph and \tcomp are impacted. This explains why we had to analyze each predicate profile separately. 
For the sake of readability, the analogous plots for the predicate profiles [5,200] and [200,400] are presented and discussed in the appendix.
Figure~\ref{fig:linear-total} shows the total runtime (\ttotal) for sets of TGDs with various sizes (\nrule). Figure~\ref{fig:linear-parse-graphcomponent-d3} breaks down \ttotal into the time to parse the TGDs (\tparse) and the time to build their dependency graph and find the special SCCs (\tgraph + \tcomp), whereas Figure~\ref{fig:linear-graph-component-d3} shows separately \tgraph and \tcomp. Figure~\ref{fig:linear-graphcomponent-d3} zooms in \tcomp.


It is evident from the above scatter plots that the time parameters \tparse and \tgraph increase linearly as long as we increase \nrule, whereas \tcomp increases very slowly. This is essentially what we have observed for simple-linear TGDs in Figure~\ref{fig:slinear-time}, with the key difference that the time needed to parse the TGDs (\tparse) is now much less compared to the time for building the dependency graph and finding the special SCCs (\tgraph + \tcomp). It should not be forgotten, however, that for linear TGDs we need to focus on a single predicate profile in order to get these linear trends; otherwise, if we consider all the predicate profiles at once, there is no trend that can be observed. Note also that the absolute running time increases compared to the case of simple-linear TGDs.

\subsection{Take-home Messages}\label{sec:discussion-linear}

The main takeaway is that the algorithm $\mathsf{IsChaseFinite[L]}$ consists of two components that are of different nature in the sense that their runtime is impacted by different parameters of the input. 
On the one hand, we have the db-dependent component, which is responsible for finding the database shapes, that is only affected by the size of the database.
On the other hand, we have the db-independent component, that is, simplify the given set of linear TGDs by using the database shapes, build the dependency graph of the simplified set of TGDs, and find the special SCCs in this graph, whose runtime is primarily affected by the number of TGDs (\nrule). Having said that, we have also observed that the number of predicates also affects the runtime of the db-independent component since it impacts the number of shapes, which in turn affects the process of dynamic simplification and the size of the dependency graph.
We conclude by observing that the total runtime of $\mathsf{IsChaseFinite[L]}$ is quite reasonable, which should be seen as a strong evidence that fast checking for the finiteness of the chase instance in the case of linear TGDs is not an unrealistic goal. Note that most of the total end-to-end runtime of the algorithm is spent on finding database shapes, which indicates that our future efforts should be concentrated on improving the db-dependent component.

\section{Validation of Results} \label{sec:rw-kb}

In Sections~\ref{sec:slinear-ex} and~\ref{sec:linear-ex}, we have presented an experimental evaluation of the algorithms $\mathsf{IsChaseFinite[SL]}$ and $\mathsf{IsChaseFinite[L]}$ using synthetically generated databases and sets of TGDs. 
In this final section, we use databases and sets of TGDs that are available in the literature with the aim of validating the main outcome of the stress test analysis performed using the synthetic scenarios. 

\subsection{Adopted Scenarios}\label{sec:kb}

We considered three families of databases and sets of TGDs that are briefly discussed below. They are also summarized in Table~\ref{tab:param}, where we report some statistics about (i) the underlying schema, i.e., number of predicates (\texttt{n-pred}) and arity of predicates (\texttt{arity}), (ii) the size of the database, i.e., number of atoms (\nfact) and number of shapes (\nshape), and (iii) the number of TGDs (\nrule).


\medskip

\noindent 
\textbf{Deep.}
This family collects sets of simple-linear TGDs that are at the same time weakly-acyclic~\cite{BKMMPST17}. It has been developed to test scenarios with a large number of chase applications and large source instances, as well as a significant number of source-to-target TGDs and target TGDs in a data exchange setting.

\medskip

\noindent 
\textbf{LUBM.}
This is a popular benchmark consisting of an ontology modelled using the central Description Logic (DL) EL, called Univ-Bench, and a data generator, called UBA, for generating synthetic data over the vocabulary of Univ-Bench~\cite{guo2005lubm}. 
We use four members of the LUBM family, namely \lubmone, \lubmtwo, \lubmthree, and \lubmfour.
Let us clarify that the DL axioms occurring in Univ-Bench can be easily converted into TGDs with one atom in the head that do not repeat variables in an atom. However, not all the axioms lead to linear TGDs, and thus, we kept only those that can be converted into linear TGDs (which are also simple-linear).

\medskip

\noindent 
\textbf{\ibench.}
 This is a framework for generating dependencies such as TGDs with tuning parameters that can control a wide range of properties~\cite{arocena2015ibench}. For our experiments, we use the following sets of simple-linear TGDs generated using \ibench: (i) STB-128, derived from an earlier STBenchmark~\cite{ATV08}, and is the smaller scenario of the family, and (ii) ONT-256, a scenario that has several times larger source instances.
For both of those sets of TGDs, we used the databases from~\cite{BKMMPST17} that were generated using the data generator in~\cite{barbosa2002toxgene} and consists of 1000 tuples per source relation. 
%

\medskip

Let us remark that in what follows we discuss our experimental evaluation of the algorithm $\mathsf{IsChaseFinite[L]}$ for linear TGDs.
Concerning the algorithm $\mathsf{IsChaseFinite[SL]}$ for simple-linear TGDs, there is not much to discuss other than the fact that it runs in a few milliseconds for all the scenarios discussed above. This is a confirmation that for simple-linear TGDs, checking for the finiteness of the chase instance can be done very efficiently.

\begin{table}[t]
	\centering
	\resizebox{\columnwidth}{!}{
		\begin{tabular}{llccccc}
			\toprule
			Family & Name & \texttt{n-pred} & \texttt{arity} & \nfact & \nshape & \nrule \\
			\midrule
			\multirow{3}{*}{Deep} & \texttt{Deep-100} & \multirow{3}{*}{1299} & \multirow{3}{*}{4} & \multirow{3}{*}{1000} & \multirow{3}{*}{1000} & 4241
			\\
			& \texttt{Deep-200} &  &  &  & & 4541
			\\
			& \texttt{Deep-300} &  &  & & & 4841
			\\\midrule
			\multirow{4}{*}{LUBM} & \lubmone & \multirow{4}{*}{104} & \multirow{4}{*}{[1,2]}  & 99547 & \multirow{4}{*}{30} & \multirow{4}{*}{137}\\
			& \lubmtwo &  &  & 1272575 &  &\\
			& \lubmthree &  &  & 13405381 &  & \\
			& \lubmfour &  &  & 133573854 & & \\
			\midrule
			\multirow{2}{*}{\ibench} & \ibenchstd & 287 & [1,10] & 1109037
			& 129 & 231\\
			& \ibenchont & 662 & [1,11] & 2146490 & 245 & 785\\
			\bottomrule
	\end{tabular}}
	\medskip
	\caption{The families Deep, LUMB, and iBench.}
	\label{tab:param}
\end{table}

\ignore{
\begin{table}[t]
\centering
\resizebox{\columnwidth}{!}{
\begin{tabular}{l|l||c|c|c|c|c}
\toprule
Family & Name & \nrule & \nfact & \texttt{n-pred} & \texttt{arity} & \nshape \\
\midrule
\multirow{3}{*}{Deep} & \texttt{deep-100} & 4,241 & \multirow{3}{*}{1,000} & \multirow{3}{*}{1,299} & \multirow{3}{*}{4} & \multirow{3}{*}{1,000} 
\\
& \texttt{deep-200} & 4,541 &  &  & & 
\\
& \texttt{deep-300} & 4,841 &  & & & 
\\\midrule
\multirow{4}{*}{LUBM} & \lubmone & \multirow{4}{*}{137} & 99,547 & \multirow{4}{*}{104} & \multirow{4}{*}{1.45 [1,2]} & \multirow{4}{*}{30}\\
& \lubmtwo &  & 1,272,575 &  &  &\\
& \lubmthree &  & 13,405,381 &  &  & \\
& \lubmfour &  & 133,573,854 &  & & \\
\midrule
\multirow{2}{*}{\ibench} & \ibenchont & 785 & 2,146,490 & 662 & 3.78 [1,11] & 245\\
& \ibenchstd & 231 & 1,109,037 & 287
 & 3.76 [1,10] & 129\\
\bottomrule
\end{tabular}}
\medskip
\caption{The families Deep, LUMB, and iBench.}
\label{tab:param}
\end{table}
}

\subsection{Experimental Evaluation} 
We run the algorithm $\mathsf{IsChaseFinite[L]}$ with all the scenarios discussed above and the experimental results are summarized in Table~\ref{tab:real}. Note that \ttotal refers to the end-to-end runtime for checking the finiteness of the chase, i.e, \ttotal = \tparse + \tgraph + \tcomp + \tshape. 
We highlight in a box the best end-to-end runtime obtained by considering either the in-memory or the in-database implementation for finding the database shapes. 

%
The parsing time (\tparse) is insignificant in all scenarios as the number of rules (\nrule), which is the main parameter that impacts the parsing time, is at most 4000. The time to build the dependency graph (\tgraph) and the time to find the special SCCs (\tcomp) are also negligible. This due to the limited number of shapes (\nshape) in these scenarios, which means that the dynamic simplification that uses those shapes does not construct a large set of simplified TGDs, and thus, the induced dependency graph is rather small. 
It is evident from Table~\ref{tab:real} that the most costly task in all scenarios is finding the database shapes, which is consistent with what we observed in Section~\ref{sec:slinear-ex}.

We report the time to find the shapes (\tshape) for both the in-memory and the in-database implementations of the procedure $\mathsf{FindShape}$. The in-memory implementation is faster for the Deep family since the underlying database has many singleton relations, i.e., relations with only one tuple, and thus, loading the tuples and finding the shapes can be done efficiently. On the other hand, the in-database implementation takes more time because it runs one query per relation, which results in many queries. 
For the LUBM and \ibench families, the in-database implementation is faster as it finds the shapes by running a few queries due to the small number of predicates. Note that, in general, the runtime of the in-database implementation increases with the number of tuples in the relations, which is evident from the LUBM scenarios. 
As a result, finding the shapes for \lubmfour takes a significant time. However, it is still much lower than the time for finding the database shapes using the in-memory implementation that requires loading many tuples.

\begin{table}[t]
\centering
\resizebox{\columnwidth}{!}{
\begin{tabular}{lccccccc}
& & & & \multicolumn{2}{c}{\textcolor{blue}{In-db}} & \multicolumn{2}{c}{\textcolor{violet}{In-memory}} \\
\toprule
Name & \tparse & \tgraph & \tcomp & \textcolor{blue}{\tshape} & \textcolor{blue}{\ttotal} & \textcolor{violet}{\tshape}  & \textcolor{violet}{\ttotal} \\
\midrule
\texttt{Deep-100} & 214 & 90 & 10 & \multirow{3}{*}{6,641} & 6,957 & 447 & \HL{763} 
\\
\texttt{Deep-200} & 265 & 116 & 9 & & 7,033 & 447 & \HL{839}
\\
\texttt{Deep-300} & 234 & 100 & 11 & & 6,986 & 500 & \HL{846}
\\\midrule
\lubmone & 84 & 10 & 1 & 221 & \HL{318} & 2,724 & 2,820\\
\lubmtwo & 46 & 10 & 1 & 830 & \HL{889} & 10,943 & 11,002\\
\lubmthree & 45 & 11 & 1 & 6,396 & \HL{6,454} & 70,131 & 70,189\\
\lubmfour & 43 & 231 & 80 & 65,578 & 
\HL{65,932} & 854,015 & 854,369\\
\midrule
\ibenchstd & 78 & 18 & 7 & 4,991 & \HL{5,096} & 7,379 & 7,484\\
\ibenchont & 179 & 35 & 8 & 11,726 & \HL{11,949} & 15,761 & 
15,984\\
\bottomrule
\end{tabular}}
\medskip
\caption{Runtime of $\mathsf{IsChaseFinite[L]}$ in milliseconds.}
\label{tab:real}
\end{table}

\subsection{Discussion}\label{sec:discussion-real}

It is fair to conclude that the experimental evaluation performed in this section confirms the main outcome of the analysis for the algorithm $\mathsf{IsChaseFinite[L]}$ performed in Section~\ref{sec:linear-ex}.
In particular, we observe that indeed the costly task is finding the database shapes, whereas the time taken by the db-independent component is negligible.
Moreover, we see that checking for the finiteness of the chase instance can be done rather efficiently in practice. In particular, for sets consisting of thousands of TGDs such as the Deep scenarios, and millions of facts such as \lubmtwo, it takes less than a second.
For schemas with a large number of predicates (e.g., \ibenchstd and \ibenchont) and databases with a large numbers of atoms (e.g., \lubmthree), it takes less than 10 seconds.
Finally, for very large databases with hundreds of millions of atoms such as \lubmfour, it takes around a minute, which is a reasonable time taking into account the actual size of the input.


Another interesting takeaway from the experimental evaluation of this section is that there is no clear way to go regarding the implementation of $\mathsf{FindShape}$ among the in-memory and the in-database options.
In particular, the in-memory implementation is preferred when there are a few tuples per relation in the database, whereas the in-database implementation performs better when the underlying schema has a few predicates of small arity. For schemas with many predicates, each of which has many tuples in the input database, both implementations require significant time, and the offline computation of the database shapes might be preferred.

\section{Conclusions and Future Work}\label{sec:conclusion}

Our work provides the first systematic attempt to experimentally evaluate algorithms devised for the semi-oblivious chase termination problem in the presence of (simple-)linear TGDs.
Our analysis revealed that for simple-linear TGDs, we can efficiently check whether the chase terminates even for very large databases and sets of TGDs.
Concerning linear TGDs, the overall runtime of the algorithm is quite reasonable, but there is still room for improvements. Interestingly, our analysis showed that the algorithm for linear TGDs consists of two separate components, the db-dependent and the db-independent components. This modular nature of the algorithm allows us to study and improve the two components separately. In particular, we have observed that the heavy component is the db-dependent one, and thus, we can focus our future efforts to improve the performance of that component. Although our analysis relied on an in-database and an in-memory implementation of the procedure for finding the shapes, we could adopt other techniques depending on the underlying application without affecting the db-independent component. An interesting direction is to materialize and incrementally keep updated the shapes in a database, which will improve the performance of the db-dependent component. 



\bibliographystyle{ACM-Reference-Format}

\newpage
\appendix
\section{Db-Independent Component} \label{sec:smaller-pred-profiles}

\begin{figure}[t]
	\centering
	\begin{subfigure}[b]{0.23\textwidth}
		\centering
		\includegraphics[width=\textwidth]{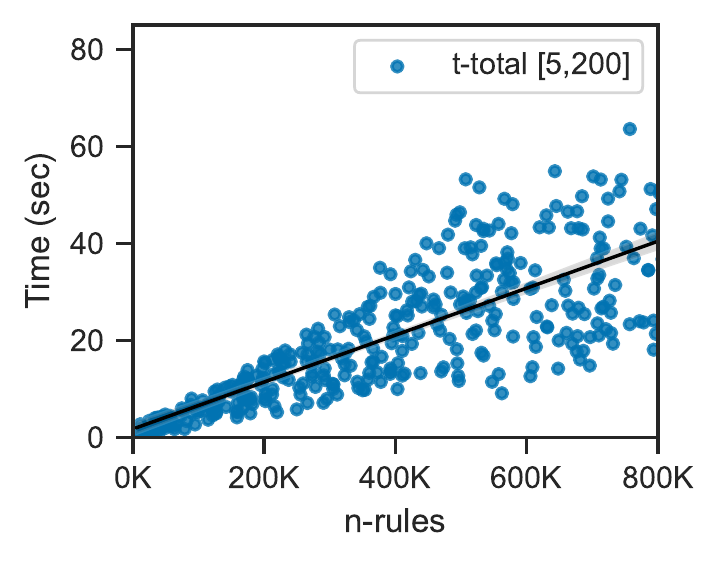}
		\caption{\ttotal}
		\label{fig:linear-total-d1}
	\end{subfigure}
	\hfill
	\begin{subfigure}[b]{0.23\textwidth}
		\centering
		\includegraphics[width=\textwidth]{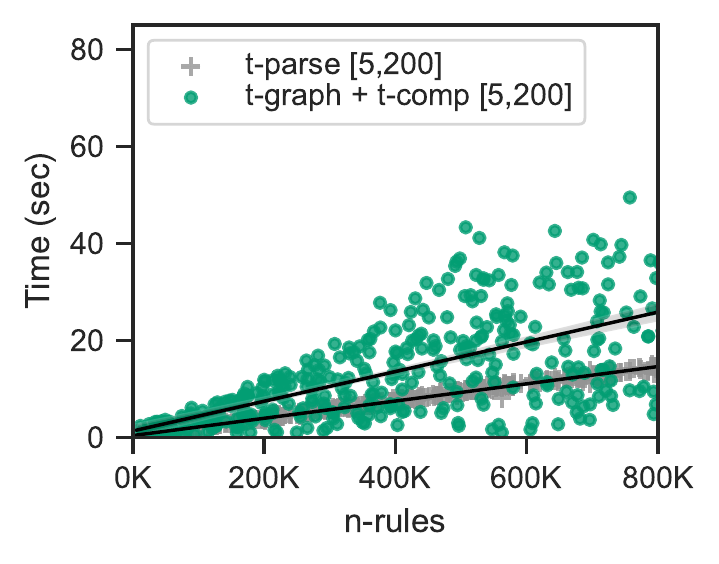}
		\caption{\tparse vs. \tgraph+\tcomp}
		\label{fig:linear-parse-graphcomponent-d1}
	\end{subfigure}
	\hfill
	\begin{subfigure}[b]{0.23\textwidth}
		\centering
		\includegraphics[width=\textwidth]{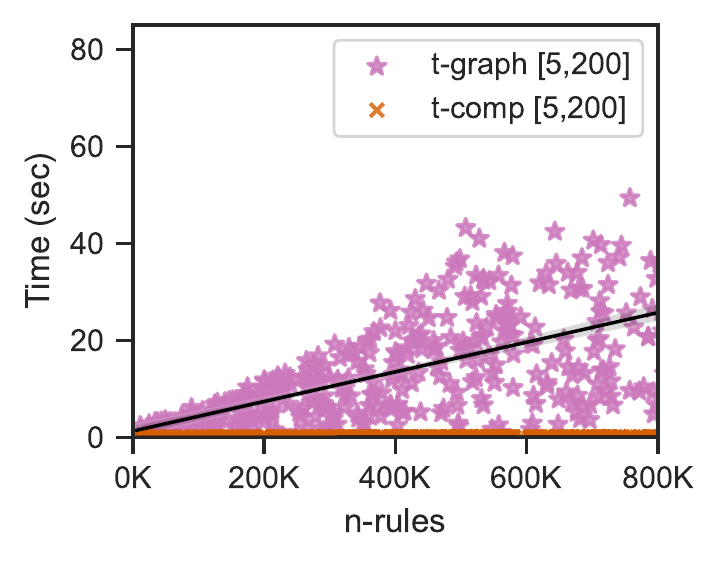}
		\caption{\tgraph vs. \tcomp}
		\label{fig:linear-graph-component-d1}
	\end{subfigure}
	\hfill
	\begin{subfigure}[b]{0.23\textwidth}
		\centering
		\includegraphics[width=\textwidth]{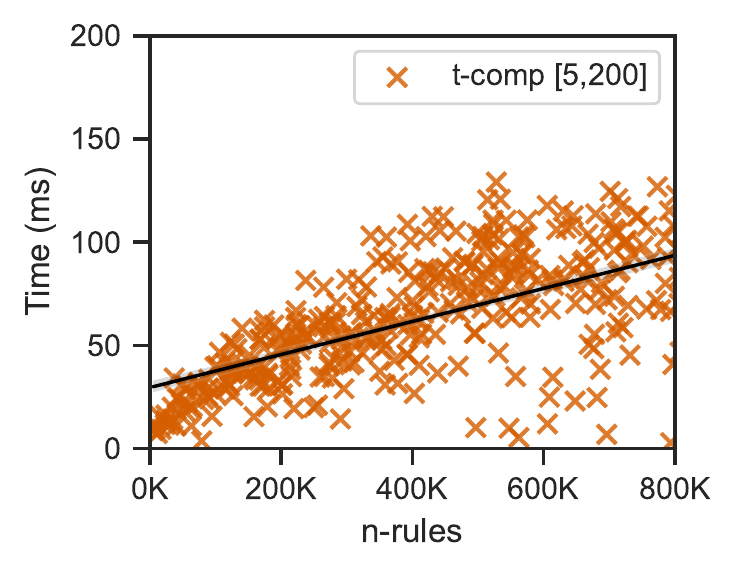}
		\caption{\tcomp}
		\label{fig:linear-graphcomponent-d1}
	\end{subfigure}
	\caption{Runtime of the db-independent component for the predicate profile [5,200].}
	\label{fig:linear-time-5-200}
\end{figure}

\begin{figure}[t]
	\centering
	\begin{subfigure}[b]{0.23\textwidth}
		\centering
		\includegraphics[width=\textwidth]{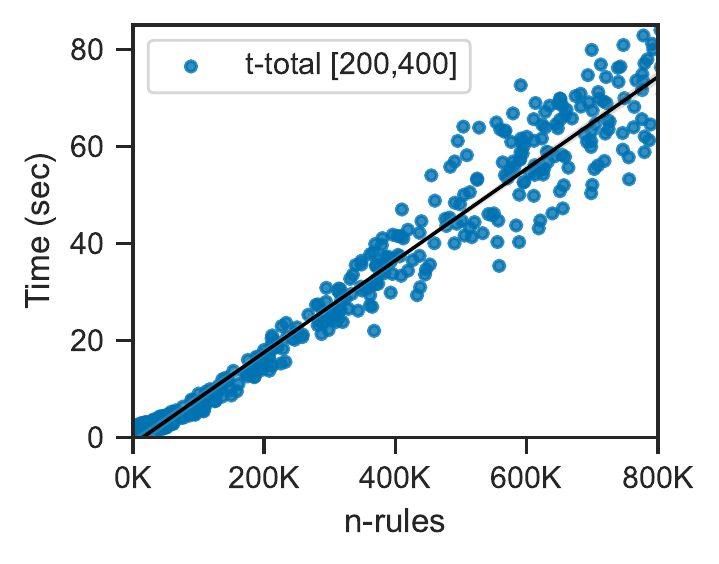}
		\caption{\ttotal}
		\label{fig:linear-total-d2}
	\end{subfigure}
	\hfill
	\begin{subfigure}[b]{0.23\textwidth}
		\centering
		\includegraphics[width=\textwidth]{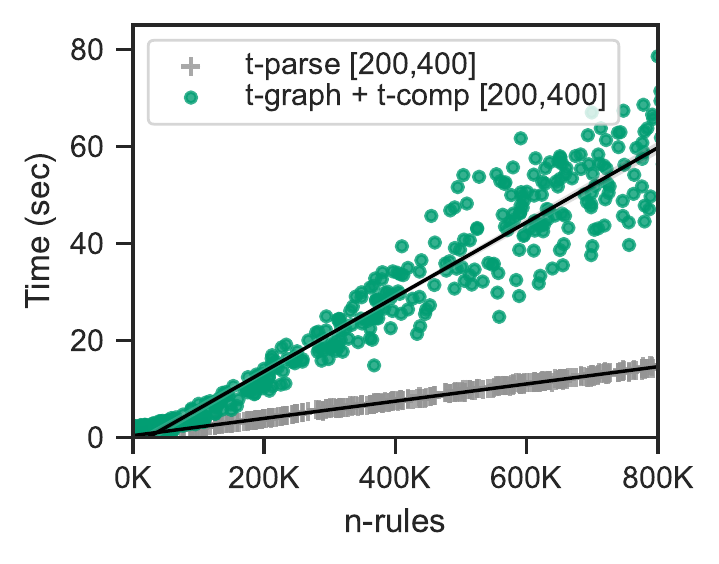}
		\caption{\tparse vs. \tgraph+\tcomp}
		\label{fig:linear-parse-graphcomponent-d2}
	\end{subfigure}
	\hfill
	\begin{subfigure}[b]{0.23\textwidth}
		\centering
		\includegraphics[width=\textwidth]{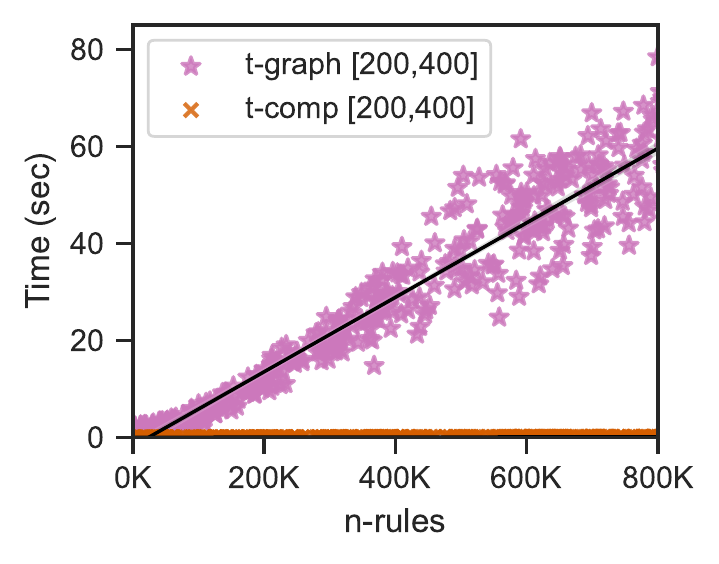}
		\caption{\tgraph vs. \tcomp}
		\label{fig:linear-graph-component-d2}
	\end{subfigure}
	\hfill
	\begin{subfigure}[b]{0.23\textwidth}
		\centering
		\includegraphics[width=\textwidth]{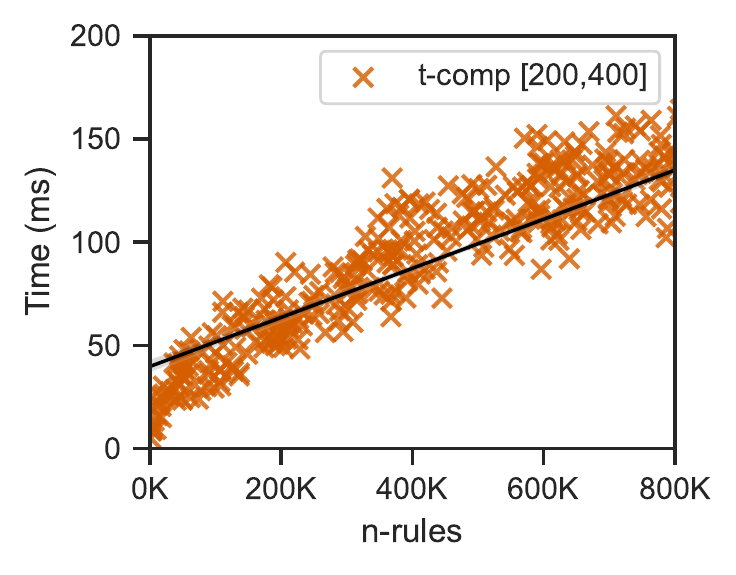}
		\caption{\tcomp}
		\label{fig:linear-graphcomponent-d2}
	\end{subfigure}
	\caption{Runtime of the db-independent component for the predicate profile [200,400].}
	\label{fig:linear-time-200-400}
\end{figure}

Recall that, for the sake of readability, in Section~\ref{sec:linear-ex} we omitted the plots that present the runtime of the db-independent component of the chase termination algorithm $\mathsf{IsChaseFinite[L]}$ for the smaller predicate profiles.
The plot for the predicate profile [5,200] is shown in Figure~\ref{fig:linear-time-5-200}, whereas for the profile [200,400] in Figure~\ref{fig:linear-time-200-400}.
We can observe, similarly to the predicate profile [400,600], that the time parameters \tparse and \tgraph increase linearly as long as we increase \nrule, whereas \tcomp increases very slowly. However, it is clear from the plots that, as long as we move to smaller predicate profiles and larger sets of TGDs, the above linear trends become less apparent. 
This phenomenon can be explained as follows. As we decrease the number of predicates that appear in the schema, it is likely that a large set of TGDs gives rise to a limited number of edges in the underlying dependency. This is because many TGDs simply lead to the same edges, which are of course considered once in the graph. This is confirmed by the following plot


\centerline{\includegraphics[width=.23\textwidth]{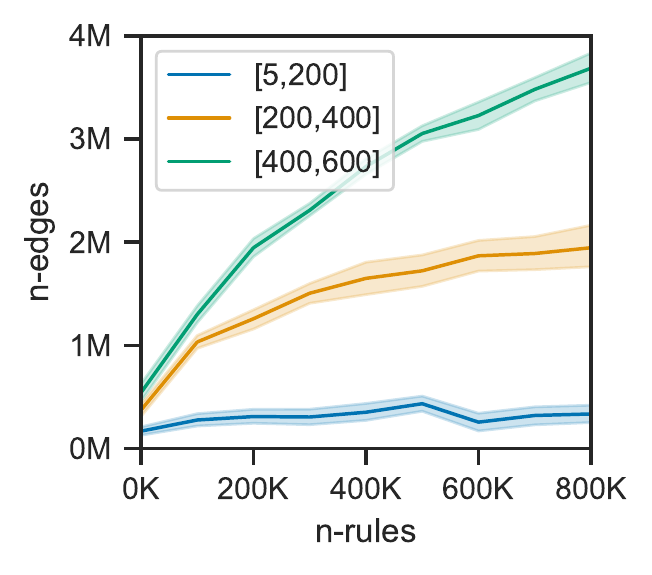}}


\noindent which depicts the average number of edges in the dependency graphs of the sets of linear TGDs used in our experimental evaluation. For the smaller predicate profiles, increasing the number of TGDs leads to a smaller increase in the graph size compared with the larger predicate profiles. In particular, for the profile [5,200], the number of edges in the dependency graphs remains almost the same as we increase the number of TGDs.

\end{document}